\documentclass[a4paper,12pt]{article}
\usepackage{JASA_manu}
\usepackage{amsmath}
\usepackage{latexsym}
\usepackage{amssymb}
\usepackage{authblk}
\usepackage{amsfonts}
\usepackage{bm}
\usepackage{color}
\usepackage{srcltx}
\usepackage{multirow}
\usepackage{graphicx}
\usepackage{natbib}
\usepackage{amsthm}
\usepackage{wasysym}
\usepackage{url}

\usepackage{caption}
\usepackage{pdflscape}
\usepackage{relsize}
\usepackage{tikz}
\usetikzlibrary{shapes.geometric,arrows,positioning,shapes, shapes.multipart}
\usepackage{graphicx}

\usepackage{amsmath,amsfonts,amsmath,bm,amssymb,amsxtra,amsthm,bm,mathtools}
\usepackage{enumitem}
\usepackage{listings}
\usepackage{setspace}
\usepackage{subfigure}
\usepackage{multirow}
\usepackage{multicol}
\usepackage{tabularx}
\usepackage{accents}
\usepackage{courier}
\usepackage{verbatim}
\usepackage{xcolor}
\usepackage [english]{babel}

\usepackage[scr=boondoxo,scrscaled=1.05]{mathalfa}
\usepackage{floatrow}
\usepackage{array, booktabs}
\usepackage[utf8]{inputenc}

\newcolumntype{L}[1]{>{\raggedright\let\newline\\\arraybackslash\hspace{0pt}}m{#1}}
\newcolumntype{C}[1]{>{\centering\let\newline\\\arraybackslash\hspace{0pt}}m{#1}}
\def\underaccent#1{\mathord{\vtop{\ialign{##\crcr
$\hfil\displaystyle{#1}\hfil$\crcr\noalign{\kern1.5pt\nointerlineskip}
$\hfil\tilde{}\hfil$\crcr\noalign{\kern1.5pt}}}}}

\newcommand{\indep}{\perp \!\!\! \perp}

\newcommand{\R}{\mathbb{R}}
\newcommand{\N}{\mathbb{N}}

\newcommand{\bh}{\mathbf{h}}
\newcommand{\bs}{\mathbf{s}}

\newcommand{\bN}{\mathbf{N}}

\newcommand{\lsi}{\lambda(\bs_i)}
\newcommand{\lsj}{\lambda(\bs_j)}
\newcommand{\E}{\mathds{E}}
\usepackage{dsfont}

\usepackage{xcolor}
\usepackage[hidelinks]{hyperref}
\hypersetup{urlcolor=blue, citecolor=blue, colorlinks=true, linkcolor=blue}

\newtheorem{theo}{Theorem}

\newtheorem{cor}{Corollary}
\newtheorem*{proof}{Proof}

\DeclareMathOperator\Corr{Corr}

\DeclareMathOperator\Var{Var}

\newcommand{\blind}{1}

\begin{document}

\if1\blind
{
\title{
Modelling Point Referenced Spatial Count Data: A Poisson Process Approach
}

\author[1,2]{Diego Morales-Navarrete}
\affil[1]{Departamento de Estad\'istica, Pontificia Universidad Cat\'olica de Chile, Santiago, Chile}
\affil[2]{Millennium Nucleus Center for the Discovery of Structures in Complex Data, Chile, \texttt{dbmorales@mat.uc.cl}}

\author[3,4]{Moreno Bevilacqua}
\affil[3]{Facultad de Ingenier\'ia y Ciencias, Universidad Adolfo Ib\'a\~nez, Vi\~na del Mar, Chile}
\affil[4]{
Dipartimento di Scienze Ambientali, Informatica e Statistica, Ca’ Foscari University of Venice, Italy, \texttt{moreno.bevilacqua@uai.cl}}

\author[5]{Christian Caama\~no-Carrillo}
\affil[5]{Departamento de Estad\'istica, Universidad del B\'io-B\'io, Concepci\'on, Chile,  \texttt{chcaaman@ubiobio.cl}}

\author[1,2,6]{Luis M. Castro}
\affil[6]{Centro de Riesgos y Seguros UC, Pontificia Universidad Cat\'olica de Chile, Santiago, Chile, \texttt{mcastro@mat.uc.cl}}

  \maketitle
} \fi

\if0\blind
{
  \bigskip
  \bigskip
  \bigskip
  \begin{center}
    {\LARGE\bf Modelling Point Referenced Spatial Count Data: A Poisson Process Approach}
\end{center}
  \medskip
} \fi

\bigskip

\newpage
\begin{abstract}
	
{\small
Random fields are useful mathematical tools  for representing natural phenomena with complex dependence structures in space and/or time. In particular, the Gaussian random field is commonly used due to its attractive properties and mathematical tractability. However, this assumption seems to be restrictive when dealing with counting data. To deal with this situation, we propose a random field with a Poisson marginal distribution considering a sequence of independent copies of a random field with an exponential marginal distribution as ``inter-arrival times'' in the counting renewal processes framework. Our proposal can be viewed as a spatial generalization of the Poisson counting process.

Unlike the classical hierarchical  Poisson Log-Gaussian model, our proposal generates a (non)-stationary random field that is mean square continuous and with Poisson marginal distributions. For the proposed Poisson spatial random field, analytic expressions for the covariance function and the bivariate distribution are provided. In an extensive simulation study, we investigate the weighted pairwise likelihood as a method for estimating  the Poisson random field parameters. 

Finally, the effectiveness of our methodology is illustrated by an analysis of reindeer pellet-group survey data, where a zero-inflated version of the proposed model is compared with  zero-inflated Poisson Log-Gaussian and Poisson Gaussian copula models. 
Supplementary materials for this article, including technical proofs and \texttt{R} code for reproducing the
work, are available as an online supplement.
}

\end{abstract}

%
%

\section{Introduction} 
\label{sec:introduction}

The faecal pellet count technique is one of the most popular tool for estimating an animal species' abundance. Specifically, this technique uses the number of observed droppings combined with their decay time and the target animal species' defecation rate. With these ingredients, it is possible to obtain an accurate density estimation of an animal population. This method was proposed by \cite{BennetEtAl1940} and has been improved by several authors \citep[see for example][among others]{VanEttenBennet1965,MayleEtAl1999,KrebsEtAl2001}.

The study motivating our research is a reindeer pellet-group survey conducted in the northern forest area of Sweden and previously analysed by \cite{Lee2016}. The objective of this survey was to assess the impact of newly established wind farms on reindeer habitat selection. This choice is crucial for the reindeer since it involves trade‐offs between fulfilling necessities for feeding, mating, parental care, and risk mitigation of predation \citep{SivertsenEtAl2016}.

Survey data was collected over the years 2009–2010 and presented a large number of zero counts. This situation is frequent when spatial species count data are collected since the survey is conducted using a point transect design \citep{BucklandEtAl2001}. This design considers a set of $K$ plots as systematically spaced plots along lines (transects) located throughout the survey region, where $K$ should be at least 20 for obtaining robust estimates of the abundance. The study area was 250 km$^2$, the distance between each transect was 300 m. On each transect, the distance between each plot was 100 m. The size of each plot was 15 m$^2$ with a radius of 2.18 m.  


From a modelling viewpoint, the analysis of the reindeer pellet-group data requires the development of statistical models for geo-referenced count data that take into account both spatial dependence and the excessive number of zeros. Random fields or stochastic processes are useful models when dealing with geo-referenced spatial or spatio-temporal data \citep{Stein:1999,Cressie:Wikle:2011,Banerjee-Carlin-Gelfand:2014}. In particular, the Gaussian random field is widely used due to  its attractive properties and mathematical tractability \citep{GELFAND201686}. Gaussianity is clearly a restrictive assumption when dealing with counting data. However, many models of current use for spatial count data employ Gaussian random fields as building blocks.

The first example is the hierarchical model approach proposed by \cite{Diggle:Tawn:Moyeed:1998}, which can be viewed as a generalized linear mixed model \citep{Diggle-Ribeiro:2007,Diggle-Giorgi:2019}. Under this framework, non-Gaussian models for spatial data can be specified using a link function and  a latent Gaussian random field through a conditionally independence assumption. In particular, the Poisson Log-Gaussian random field (Poisson LG hereafter) has been widely applied for modelling count spatial data \citep[see for instance][for interesting applications and in-depth study of its properties] {Christensen:Waagepetersen:2002,Guillot_et_all:2009,Oliveira:2013}. 
Similar models, that can be defined hierarchically in terms of the specification of the first two moments and a correlation function have been proposed in \cite{Monestiez_et_al2006} and \cite{Oliveira:2014}.
 
It is important to stress that the conditional independence assumption underlying these kind of models leads to (a) random fields with  marginal distributions that are not Poisson and (b) random fields with a ``forced'' nugget effects that implies no mean square continuity.

\begin{figure}[htb!]
\scalebox{0.7}{
\begin{tabular}{cc}
\includegraphics[width=7cm, height=5.5cm]{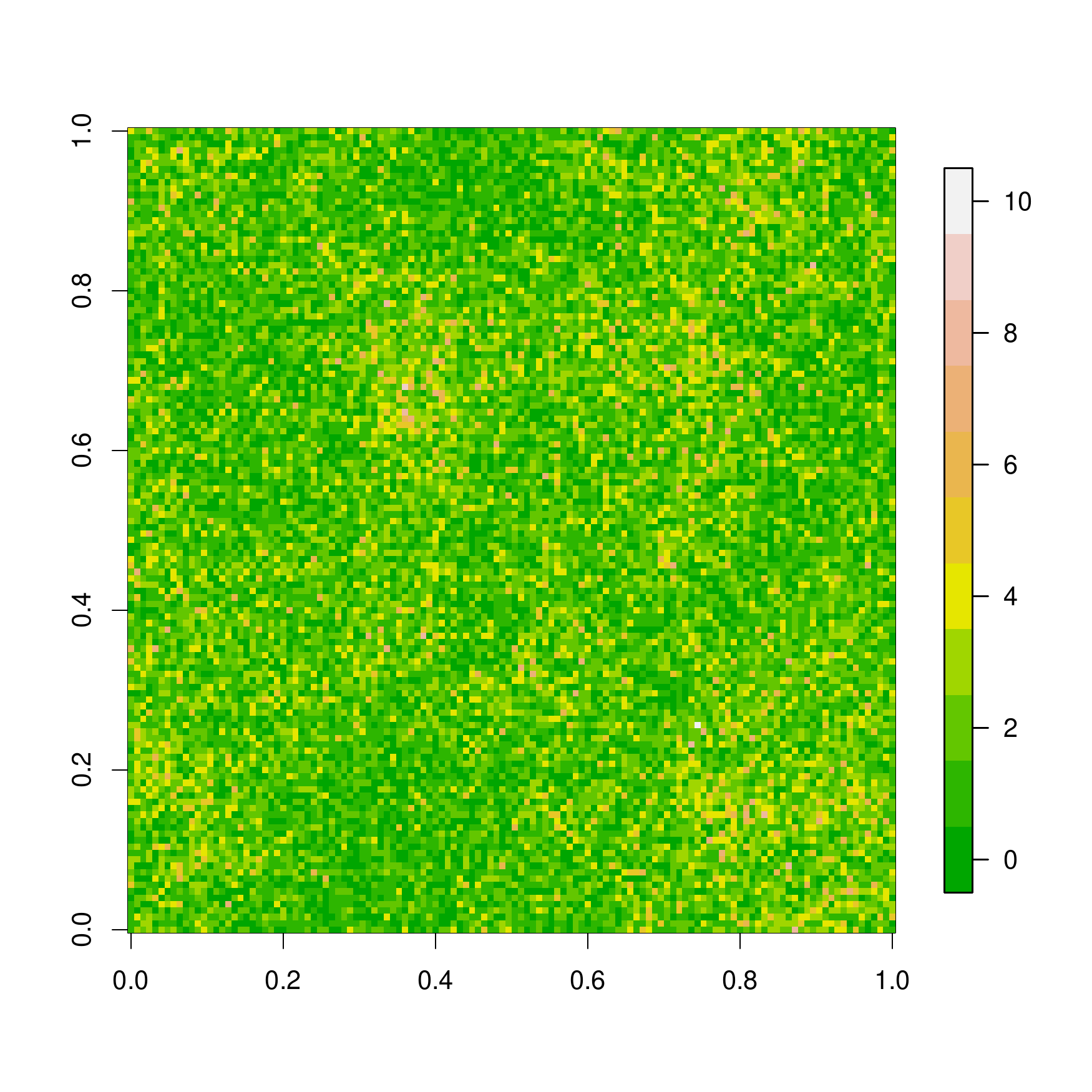} & \includegraphics[width=7cm, height=5.5cm]{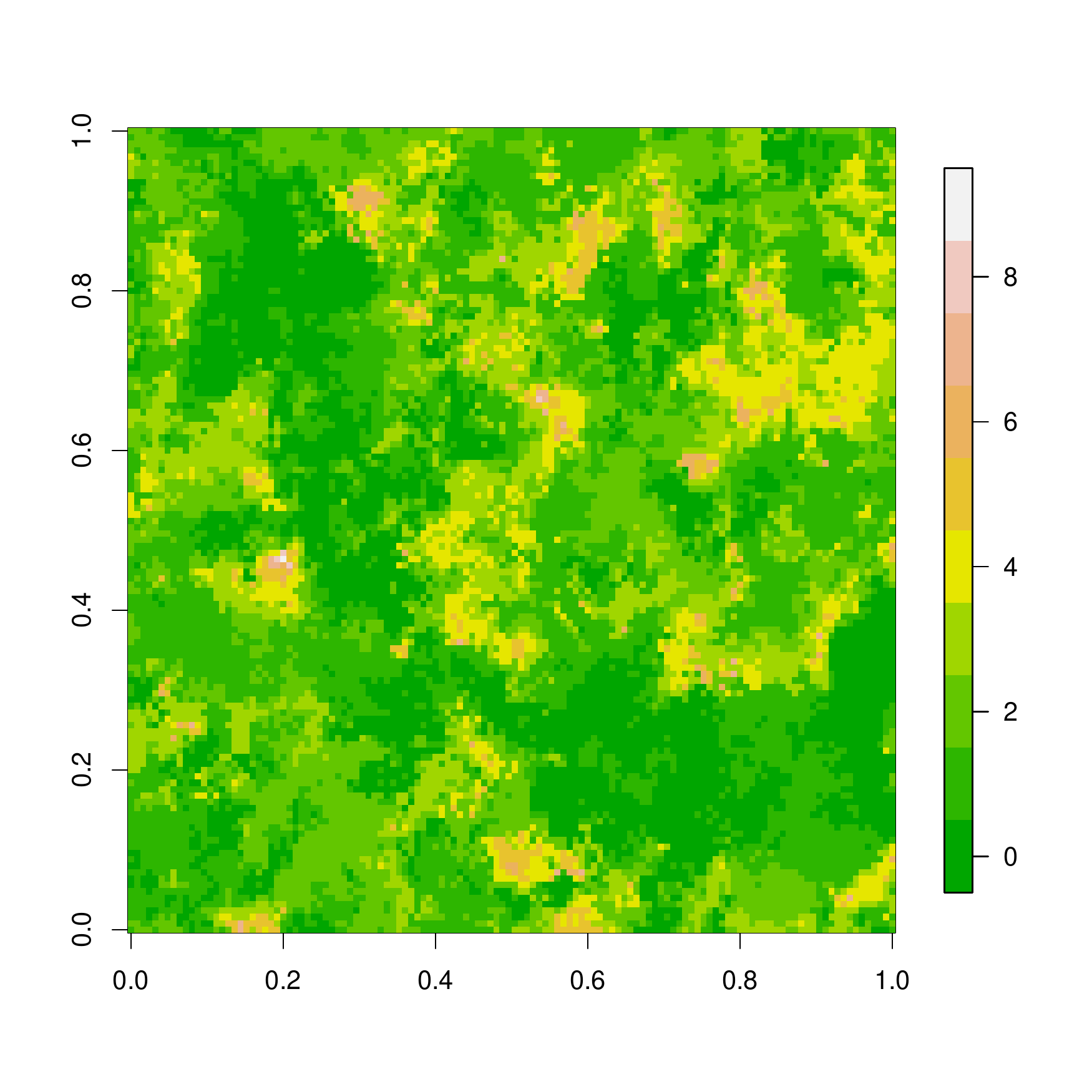}\\
(a)&(b)\\
\includegraphics[width=7cm, height=5.5cm]{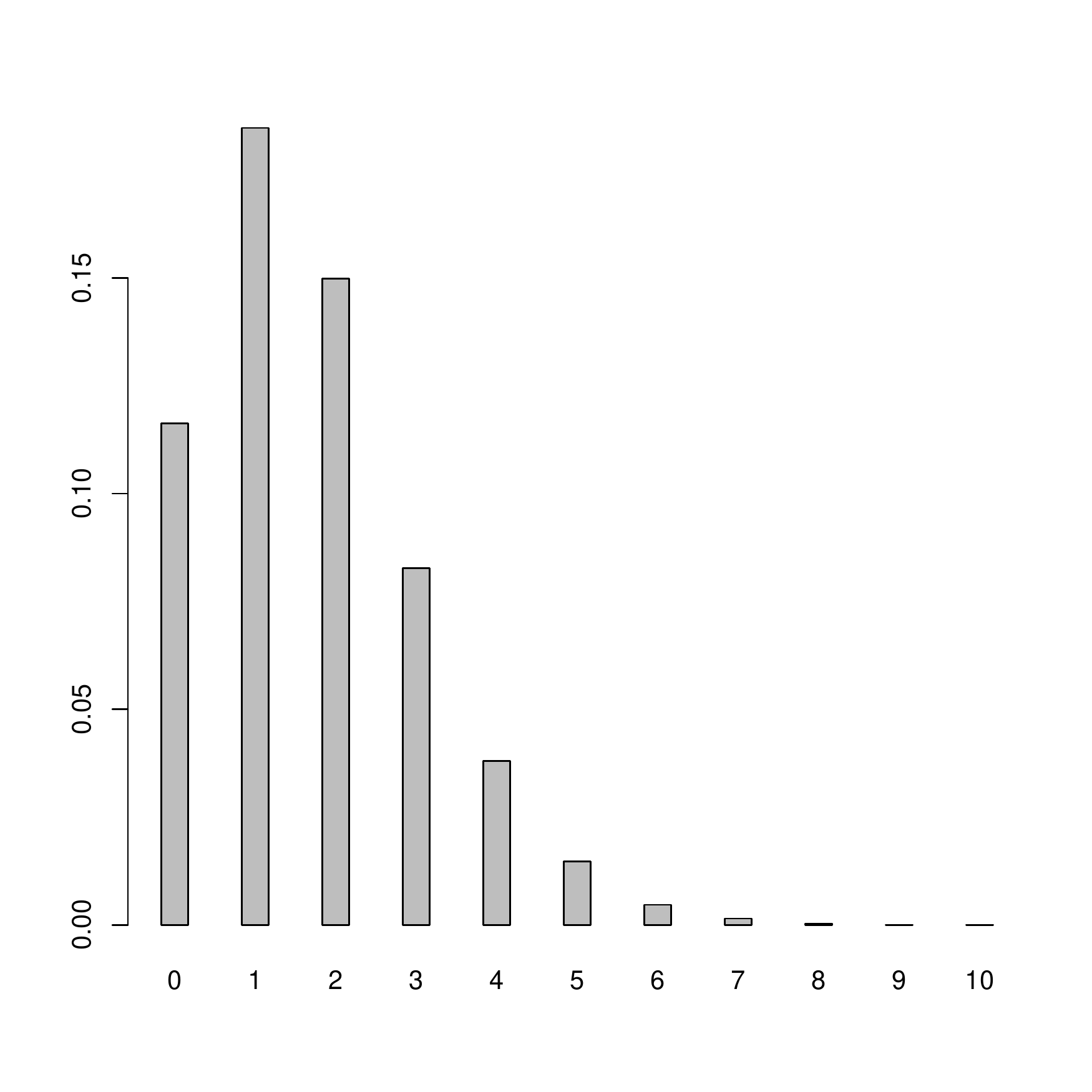} & \includegraphics[width=7cm, height=5.5cm]{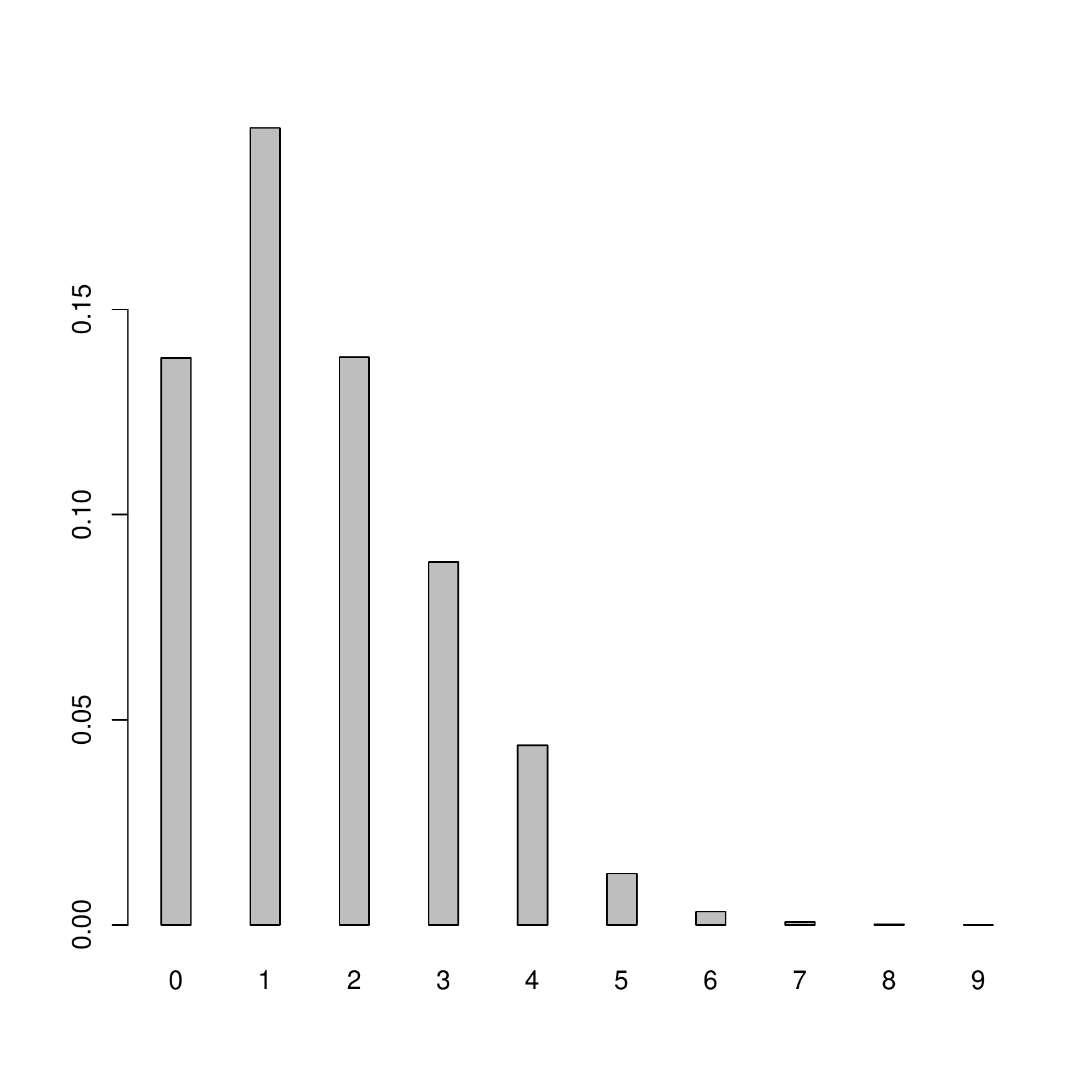} \\
(c)&(d)\\
\end{tabular}}
\caption{A realization of a Poisson LG random field,
 where the LG random field is given by $e^{\mu+\sqrt{\sigma^2} G(\bs)}$, where $G$ is a standard Gaussian random field with parameters $\mu=0.5$ and $\sigma^2=0.05$ (panel a) and its associated histogram (panel c).
A realization of our proposed Poisson random field with $\lambda=e^{0.5+0.05/2}$ (panel b) and its associated histogram (panel d). In both cases the underlying isotropic correlation is $\rho(r)=(1-r/0.5)^4_+$.
 } \label{fig:poi}
\end{figure}

To illustrate this situation, Figure \ref{fig:poi} (a) shows a realization on the unit square of a Poisson LG random field, $i.e.$  $e^{\mu+\sqrt{\sigma^2} G(\bs)}$, where $G$ is the standard Gaussian random field with isotropic correlation $\rho(r)=(1-r/0.5)^4_+$ belonging to the Generalized Wendland family \citep{Bevilacqua:Faouzi_et_all:2019}, $\mu=0.5$, $\sigma^2=0.05$, $r$ is the spatial distance and $(\cdot)_+$ denotes the positive part. In this case, the mean of the Poisson LG field is given by $\lambda=e^{0.5+0.05/2}$. The associated histogram is shown in panel (d).

Additionally, Figure \ref{fig:poi} (b) shows a realization and the associated histogram of our proposed random field (see Equation \eqref{qqq}), with the same mean and underlying correlation function of the Poisson LG model. A quick analysis of both figures reveals a ``whitening'' effect on the Poisson LG random field's paths because of the ``forced'' discontinuity at the origin of the correlation function for the Poisson LG (see Section \ref{sec:3.1}). This potential problem, which has also been highlighted by \cite{Oliveira:2013}, indicates that the Poisson LG random field may impose severe restrictions on the correlation structure and may be inadequate to model spatial count data specially consisting of small counts.

The second example is the Poisson spatial model obtained using Gaussian copula \citep{ Kazianka:Pilz:2010,Masarotto:Varin:2012,Joe:2014},
which is referred to as the  Poisson GC random field hereafter. This approach has some potential benefits with respect to the hierarchical models \citep[see][for a comparison between these two approaches]{Han:Oliveira:2016}.
For example, the resulting random field has Poisson marginals and can or cannot be mean square continuous depending on whether the latent Gaussian random field is mean square continuous or not. In addition to some some criticisms concerning the lack of uniqueness of the copula when applied to discrete data \citep{Genest:Neslehova:2007,Trivedi:Zimmer:2017}, this approach has no clarity regarding what underlying physical mechanism is generating the data, making it less interesting  from an interpretability perspective.

Our proposal tries to solve the drawbacks of the Poisson LG and of the Poisson GC approaches by specifying a new class of spatial counting random fields based on the Poisson counting process \citep{Cox:1970,MAINARDI2007725,ross2008stochastic} applied to the spatial setting. Specifically, we first consider a random field with exponential marginal distributions obtained as a rescaled sum of two independent copies of an underlying  standard Gaussian random field. Then, by considering independent copies of the exponential random field as inter-arrival times in the counting renewal processes framework, we obtain a (non-)stationary random field with Poisson marginal distributions. By construction, for each spatial location, the proposed model is a Poisson counting process \textit{i.e.} it represents the random number of events  occurring in an arbitrary interval of time when the time between the occurrence of two events is exponentially distributed. 
More importantly, given two location sites, the associated Poisson counting processes are spatially correlated. For this reason, the proposed model can be viewed as a spatial generalization of the Poisson process.

For the novel Poisson random field, we provide the covariance function and analytic expressions for the bivariate distribution in terms of the regularized incomplete Gamma and confluent hypergeometric functions \citep{Gradshteyn:Ryzhik:2007}. 
It follows that the dependence of the proposed Poisson random field is indexed by the correlation function of the underlying Gaussian random field and by the mean parameter. It is important to stress that our theoretical results are inspired by the two-dimensional renewal theory described in \cite{Hunter:1974}.

The Poisson random field estimation is performed with the weighted pairwise likelihood ({\it wpl}) method \citep{Lindsay:1988,Varin:Reid:Firth:2011,Bevilacqua:Gaetan:2015} exploiting the results obtained from the bivariate distribution. In particular, in an extensive simulation study, we explore the efficiency of the {\it wpl} method when estimating the parameters of the proposed Poisson random field. We also explore the statistical efficiency of a Gaussian misspecified version of the {\it wpl} method \citep{cppp,Bev:2020}, which is also called Gaussian quasi-likelihood in some literature \citep{AOS1121}. The
findings show that the misspecified {\it wpl} leads to a less efficient  estimator, in particular for low counts. However, the method has some computational benefits. In addition, we compare the performance of the optimal linear predictor under the proposed model with the optimal predictors obtained using the Poisson GC and Poisson LG models.

Finally, in the real data application, 
we consider a zero-inflated version of the proposed Poisson random field 
to deal with excess zeros in the reindeer pellet-group counts data, using the zero-inflated Poisson LG and Poisson GC models as benchmarks. The methods proposed in this paper are  implemented in the \texttt{R} \citep{R2020} package \texttt{GeoModels} \citep{Bevilacqua:2018aa} and \texttt{R} code  for reproducing the work is available as an online supplement.

The remainder of the paper is organized  as follows. In Section \ref{sec:2}, we provide some basic notation and describe the exponential random field. Section \ref{sec:3} introduces our proposal by presenting a new class of counting random fields under the general renewal counting framework, focusing on the Poisson random field. A study of the associated correlation function and bivariate distribution is presented and, in addition, a zero-inflated extension of the proposed model is introduced. In Section \ref{sec:4}, the {\it wpl} method for obtaining the {\it wpl} estimates and the optimal linear prediction is discussed. Section \ref{sec:5} provides an in-depth simulation study to investigate the performance of the Poisson random field in spatial and spatio-temporal settings. In Section \ref{sec:6}, the  faecal pellet-group counts dataset previously described is re-analysed. Section \ref{sec:7} closes the paper with a discussion of our main findings and future research directions.
 
\section{A random field with  exponential marginal distributions}\label{sec:2}

To make the paper self-contained, we start by  introducing some notation in this Section. For the rest of the paper, given a second order real-valued random field $Q=\{Q(\bs), \bs \in A\subset \R^d \}$, we denote by $f_{Q(\bs)}$ and  $F_{Q(\bs)}$ the marginal probability density function ({\it pdf}) and cumulative distribution function ({\it cdf}) of $Q(\bs)$, respectively. Moreover, for any set of distinct points $(\bs_1,\ldots,\bs_l)^\top$, $l\in \N$ and $\bs_i \in A$, we denote the correlation function by $\rho_Q(\bs_i,\bs_j)=\Corr(Q(\bs_i),Q(\bs_j))$. In the stationary case, the adopted notation is $\rho_Q(\bh)=\Corr(Q(\bs_i),Q(\bs_j))$, where $\bh=\bs_i-\bs_j $ is the lag separation vector. Finally, $f_{\bm{Q}_{ij}}$ denotes the {\it pdf} of the bivariate random vector $\bm{Q}_{ij}=(Q(\bs_i),Q(\bs_j))^\top$, $i \neq j$. If the random field $Q$ is a discrete-valued random field, then $\Pr(Q(\bs)=q)$ and $\Pr(Q(\bs_i)=n,Q(\bs_j)=m)$, $q,m,n \in \N$ will denote the marginal and bivariate discrete probability functions, respectively.

Let $G=\{G(\bs), \bs \in A \}$ 
be a zero mean and unit variance weakly stationary Gaussian random field with correlation function $\rho_G(\bh)$. Henceforth, we call $G$ the Gaussian underlying random field, and with some abuse of notation, we set $\rho(\bh):=\rho_{G}(\bh)$, denoting this as the underlying correlation function. Let $G_1,G_2$ be two independent copies of $G$ and let us define the random field  $W=\{W(\bs), \bs\in A\}$ as follows:
\begin{equation}\label{gamma}
   W(\bs) :=  \frac{1}{2\lambda(\bs)}\sum_{k=1}^{2} G^2_k({\bs}),
\end{equation}
where $\lambda(\bs)>0$ is a non-random function. $W$ is a stationary random field with a marginal exponential distribution, with parameter $\lambda(\bs)$ denoted by $W(\bs) \sim  \mbox{Exp}(\lambda(\bs))$ with $\E(W(\bs))=1/\lambda(\bs)$, $\mbox{Var}(W(\bs))=1/\lambda^2(\bs)$, and it can  be easily observed  that $\rho_{W}(\bh) =\rho^2(\bh)$.

The associated multivariate exponential density was discussed earlier by \cite{Krishnamoorthy:Parthasarathy:1951}, and since then, its properties have been studied by several authors \citep{Krishnaiah:Rao:1961,Royen:2004}. However, likelihood-based methods for exponential random fields can be troublesome since the analytical expressions of the multivariate density can be derived only in some specific cases. For example, when $d=1$ and the underlying correlation function is exponential, the multivariate {\it pdf} is given by \citep{Bevilacqua:2018ab}:
\vspace{-0.5cm}
\begin{eqnarray*}\label{gammafd1}
f_{W}(w_1,\ldots,w_n)&=&\exp{\left[-\frac{w_1\lambda_1}{(1-\rho^2_{1,2})}
-\frac{w_n \lambda_n}{(1-\rho^2_{n-1,n})}-\sum\limits_{i=2}^{n-1}\frac{(1-\rho^2_{i-1,i}\rho^2_{i,i+1})\lambda_i w_i}{(1-\rho^2_{i-1,i})(1-\rho^2_{i,i+1})}\right]}\nonumber\\
&&\times\prod\limits^{n-1}_{i=1}I_{0}\left(\frac{2\rho_{i,i+1}\sqrt{w_i \lambda_i w_{i+1} \lambda_{i+1}}}{(1-\rho^2_{i,i+1})}\right)   \times    \left(    \prod\limits^{n-1}_{i=1}(1-\rho^2_{i,i+1})    \right)^{-1},
\end{eqnarray*}
with
$\rho_{ij}:=\exp\{-|s_i-s_j|/\phi\}$, $\lambda_i=\lambda(s_i)$, $\phi>0$ and $I_{a}(x)$ being the modified Bessel function of the
first kind of order $a$. Regardless of the dimension of the space $A$ and the type of correlation function, 
the bivariate exponential {\it pdf} is given by \citep{Kibble:1941,VereJones:1997}:
\begin{equation*}\label{pairchi2}
f_{W_{ij}}(w_{i},w_j)=\frac{e^{-\frac{(\lambda(\bs_i)w_i+\lambda(\bs_j)w_j)}{(1-\rho^2(\bh))}}}{(1-\rho^2(\bh))}
 I_{0}
\left( \frac{ 2\sqrt{\rho^2(\bh)\lambda(\bs_i)\lambda(\bs_j)w_iw_j} } {(1-\rho^2(\bh))}\right).
\end{equation*}
The exponential random field $W$ will be further used to define a new random field with Poisson marginal distributions.

\section{Spatial Poisson random fields}\label{sec:3} 

Our proposal relies on considering an infinite sequence of independent copies $Y_1,Y_2 \ldots$, of $Y= \{Y(\bs), \bs\in A\}$, a positive continuous random field. We define a new class of counting random fields, $N_{t(\bs)}:= \{N_{t(\bs)}(\bs), \bs\in A\}$, $t(\bs)\geq 0$, as follows:
\begin{equation}\label{qqq}
N_{t(\bs)}(\bs):=
	\begin{cases}0
	&\mbox{if} \quad 0\leq t(\bs)<S_1(\bs) \\
	\max\limits_{n\geq 1}\{  S_n(\bs)\leq  t(\bs)\} & \mbox{if} \quad    S_1(\bs) \leq t(\bs) \end{cases},
\end{equation}	
where $S_n(\bs)=\sum_{i=1}^{n}Y_i(\bs)$ is the $n$-fold convolution of $Y$ and $N_{t(\bs)}(\bs)$ represents the random total number of events that have occurred up to time $t(\bs)$ at location site $\bs \in A$. In our approach, we are assuming 
that the location sites share a common time $i.e.$
$t=t(\bs)$. This assumption is justified since the observation window is fixed for each location site in the pellet group count application.

The proposed  model  can be viewed as a spatial generalization of the renewal counting processes \citep{Cox:1970,MAINARDI2007725}, where we consider independent copies of a positive random field as inter-arrival times or waiting times, instead of an independent and identically distributed sequence of positive random variables.

For each $\bs \in A$, and using the classical results from the renewal counting processes theory, the marginal discrete probability function of $N$ is given by:
\begin{equation}\label{www}
\Pr(N_t(\bs)=n)=F_{S_n(\bs)}(t) - F_{S_{n+1}(\bs)}(t).
\end{equation}
In addition, the marginal mean (the so-called renewal function) and the variance of $N_t$ are given, respectively, by:
\begin{equation*}
 \E(N_t(\bs))=\sum_{i=1}^{\infty}  F_{S_i(\bs)}(t) , \quad \mbox{Var}(N_t(\bs))=\left(2\sum_{i=1}^{\infty} i F_{S_i(\bs)}(t) - \E(N_t(\bs))\right)   -( \E(N_t(\bs)))^2.
\end{equation*}
Different elections of the positive random field $Y$ lead to counting random fields with specific marginal distributions. 

In this paper, we assume that the ``spatial inter-arrival times” are exponentially distributed \textit{i.e.} we assume $Y\equiv W$, where $W$ is the positive random field defined in \eqref{gamma},
with $\mbox{Exp}(\lambda(\bs))$ marginal distribution and {\it cdf} given by $F_{Y(\bs)}(x)=1-e^{-\lambda(\bs)x}$, $x>0$. In this case, $S_n(\bs)\sim \mbox{Gamma}(n,\lambda(\bs))$ with $n\in \N$ is an Erlang distribution, with {\it cdf} given by: 
\begin{equation*}
F_{S_n(\bs)}(x)=1-\sum_{k=0}^{n-1}\frac{e^{-\lambda(\bs)x}(\lambda(\bs)x)^{k}}{k!},\quad x>0
\end{equation*}
and from \eqref{www}, we can obtain the marginal distribution of $N$ as:
\vspace{-0.5cm}
\begin{equation}\label{kjk}
\Pr(N_t(\bs)=n)=e^{-t\lambda(\bs)} [t\lambda(\bs)]^n/n!, \quad n=0,1,2,\ldots. 
\end{equation}
with $\E(N_t(\bs))=\mbox{Var}(N_t(\bs))=t\lambda(\bs)$. 

By construction, for each $\bs$, the proposed model is a Poisson counting process that is $N_t(\bs)\sim \mbox{Poisson}(t\lambda(\bs))$ represents 
the random number of events that have occurred up to time $t$, when the inter-arrival times are exponentially distributed.
In addition, given two arbitrary location sites $\bs_1$,$\bs_2$, the associated Poisson counting processes $N_t(\bs_1)$ and $N_t(\bs_2)$ are spatially correlated (see Section \ref{sec:3.1}).
Therefore, the proposed model can be viewed as a spatial generalization of the Poisson counting process. Hereafter and without loss of generality, $t$ is set to one, and $N_t$ is denoted as $N$.

We will call $N$ a Poisson random field with underlying correlation $\rho(\bh)$
because $N$ is marginally Poisson distributed and the dependence is indexed by a correlation function.

Note that, 
when the spatially varying mean (and variance) $\lambda(\bs)$ is not constant then $N$ is not stationary.
A typical parametric specification for the mean is given by $\lambda(\bs)=e^{X(\bs)^\top\bm{\beta}}$, where $X(\bs) \in  \R^k$ is a vector of covariates and $\bm{\beta}  \in  \R^k$ even though other types of parametric and non-parametric specifications can be used.

It is important to note that although the proposed Poisson random field is defined on the $d$-dimensional Euclidean space $A$, the proposed method can be easily adapted to other spaces, such as the continuous space-time space, the spherical spaces or the discrete space. The key for this extension is the specification of a suitable underlying correlation function $\rho(\bh)$. For example, a correlation function defined on the space-time setting, {\it i.e.}, $A\subset \R^d\times \R$ \citep{Gneiting:2002} or on the sphere of arbitrary radius {\it i.e.}, $A\subseteq  \mathbb{S}^2=\{ \bs \in \R^3, ||\bs|| = M\}$, $M>0$ \citep{gneiting2013,porcubev}. In the case of lattice or areal data, a suitable precision matrix with an appropriate neighborhood structure should be specified for the underlying Gaussian Markov random field \citep{Rue:Held:2005}.
 
\subsection{Correlation function}\label{sec:3.1}

The following result, which can be found in  the pioneering work of \cite{Hunter:1974}, provides the correlation function $\rho_N(\bs_i,\bs_j)$ of the non-stationary Poisson random field with underlying correlation $\rho(\bh)$ depending on the regularized lower incomplete Gamma function:
\vspace{-0.5cm}
\begin{equation}\label{incgamma}
\gamma^{*}(a,x)=\frac{\gamma(a,x)}{\Gamma(a)}=\frac{1}{\Gamma(a)}\int\limits^{x}_{0}t^{a-1}e^{-t}dt,
\end{equation}
where $\gamma(\cdot,\cdot)$ is the lower incomplete Gamma function and $\Gamma(\cdot)$ is the Gamma function. Additionally, we define the function $\gamma^{\star}\left(a,x,x'\right)=\gamma^{\ast}\left(a,x\right)\gamma^{\ast}\left(a,x' \right)$, which considers the product of two regularized lower incomplete Gamma function sharing a common parameter.

\begin{theo}\label{Theo1}
Let $N$ be a non-stationary Poisson random field with underlying correlation $\rho(\bh)$. Then,
\begin{equation*}\label{ccp}
\rho_N(\bs_i,\bs_j)=\dfrac{\rho^2(\bh)(1-\rho^2(\bh))}{\sqrt{\lambda(\bs_i)\lambda(\bs_j)}}\sum\limits_{r=0}^{\infty}\gamma^*\left(r+1,\dfrac{\lambda(\bs_i)}{1-\rho^2(\bh)},  \dfrac{\lambda(\bs_j)}{1-\rho^2(\bh)}\ \right),
\end{equation*}
with $\bh=\bs_i-\bs_j$.
\end{theo}
\begin{proof}
For details, refer to \cite{Hunter:1974}, Section 5.2 (pages 38-39).
\end{proof}

The following result  provides a  closed form expression of the correlation function in terms of modified Bessel function  for the stationary case.

\begin{cor}\label{statcoor}
In Theorem \ref{Theo1}, when $\lambda(\bs)=\lambda$, the Poisson random field  is weakly stationary with the correlation function given by:
\begin{equation}\label{cpoi}
\rho_N(\bh,\lambda)=\rho^2(\bh)\left[1- \exp\left(-z(\bh,\lambda)\right)\left(I_0\left(z(\bh,\lambda)\right)+I_1\left(z(\bh,\lambda)\right)\right)\right],
\end{equation}
where $z(\bh,\lambda)=2\lambda (1-\rho^2(\bh))^{-1}$.
\end{cor}

\begin{proof}
See the online supplement.
\end{proof}

Note that $\rho_N(\bh)$
is well defined at the origin since 
$\lim_{\bh \to \bm{0}} \rho_N(\bh)=1$
implying that the Poisson random field is  mean square continuous. Additionally, if $\rho(\bh)= 0$, then $\rho_N(\bh)=0$ and if $\lambda \rightarrow  \infty$ then $\rho_N(\bh)= \rho^2(\bh)$, {\it i.e.}, it converges to the correlation function of an exponential random field.

Following the graphical example given in Figure \ref{fig:poi}, we now compare the correlation functions of the proposed Poisson random field with the correlation of the Poisson LG random field, which is defined hierarchically by first  considering a LG random field $Z=\{Z(\bs), \bs \in A \}$ defined as  $Z(\bs)=e^{\mu+\sigma^2G(\bs)}$, where $G$ is a standard Gaussian random field with correlation $\rho(\bh)$, and then assuming $Y(\bs)\mid Z(\bs)\sim \mbox{Poisson}(Z(\bs))$ with $Y(\bs_i) \indep  Y(\bs_j) \mid Z$ for $i\neq j$. In this case, the first two moments of $Y(\bs)$ are given by
 $\E(Y(\bs))=e^{\mu+0.5\sigma^2}$ and $\mbox{Var}(Y(\bs))=\E(Y(\bs))(1+ \E(Y(\bs))(e^{\sigma^2}-1))$. Consequently, and following \cite{Aitchison}, the correlation function is given by $\rho_Y(\bh,\mu,\sigma^2)=1$ if $\bh=\mathbf{0}$ and
\begin{equation*}
\rho_Y(\bh,\mu,\sigma^2)=\frac{e^{\sigma^2 \rho(\bh)}-1}{e^{\sigma^2}-1+\E(Y(\bs))^{-1} },
\end{equation*}  
otherwise.
This correlation is discontinuous at the origin and the nugget effect is given by: 
$$\dfrac{\E(Y(\bs))^{-1}}{\E(Y(\bs))^{-1} +e^{\sigma^2 }-1}>0.$$
It is apparent that the marginal mean $\E(Y(\bs))$ has a strong impact on the nugget effect.

Figure \ref{covaa5} (a) depicts the correlation functions $\rho_Y(\bh,0.5,0.05)$, $\rho_Y(\bh,2.5,0.1)$, and $\rho_Y(\bh,4.5,0.2)$, which correspond to Poisson LG random fields with mean $\E(Y(\bs))=1.69,$ $12.81,$ and, $99.48$, respectively. As underlying correlation model we assume $\rho(\bh)=(1-||\bh||/0.5)^{4}_+$. It can be appreciated that for large mean values, the nugget effect is negligible. However, for small mean values, the nugget effect can be huge, and it is the cause of the ``whitening" effect observed in Figure \ref{fig:poi} (a).

Figure \ref{covaa5} (b) depicts the correlation function $\rho_N(\bh,\lambda)$ of the proposed Poisson random field using the same means and underlying correlation function
of the Poisson LG random field. It can be appreciated that the correlation  covers
the entire range between 0 and 1, irrespective of the mean values.

Finally, Figure \ref{covaa5} (c) depicts  the correlation function of the Poisson GC random field \citep{Han:Oliveira:2016}   $C=\{C(\bs), \bs \in A \}$
 defined as $C (\bs) = F_{\bs}^{-1}(\Phi(G(\bs)),\lambda)$, where $\Phi(\cdot)$ is the {\it cdf} of the standard Gaussian distribution and $F_{\bs}^{-1}(\cdot,\lambda  )$ is the quantile function of the Poisson distribution and $G$ is a standard Gaussian random field with correlation $\rho(\bh)$. The correlation function in this case is given by
\begin{equation*}
\rho_C(\bh,\lambda)=\int_{-\infty}^{\infty} \int_{-\infty}^{\infty} \lambda^{-1}F_{\bs_i}^{-1}(\Phi(z_i),\lambda)F_{\bs_j}^{-1}(\Phi(z_j),\lambda) \phi_2(z_i,z_j,\rho(\bh))dz_i dz_j-\lambda,
\end{equation*}
 where $\phi_2$ is the {\it pdf} of the bivariate standard Gaussian distribution.
 It is apparent that the Poisson GC correlation $\rho_C(\bh,\lambda)$ is much stronger than  $\rho_N(\bh,\lambda)$, and it does not seem to be affected by the different mean values.

It is important to stress that the Poisson and the Poisson GC random fields that are not mean-square continuous can be obtained by introducing a nugget effect, {\it i.e.}, a discontinuity at the origin of $\rho_{N}(\bh)$. This can be achieved by replacing the underlying correlation function
$\rho(\bh)$ with $\rho^*(\bh)=\rho(\bh)(1-\tau^2)+\tau^2\mathds{1}_{0}(||\bh||)$, where $0\leq\tau^2<1$ represents the underlying nugget effect.
 
\begin{figure} 
\scalebox{0.8}{
\begin{tabular}{ccc}
\includegraphics[width=0.3\textwidth]{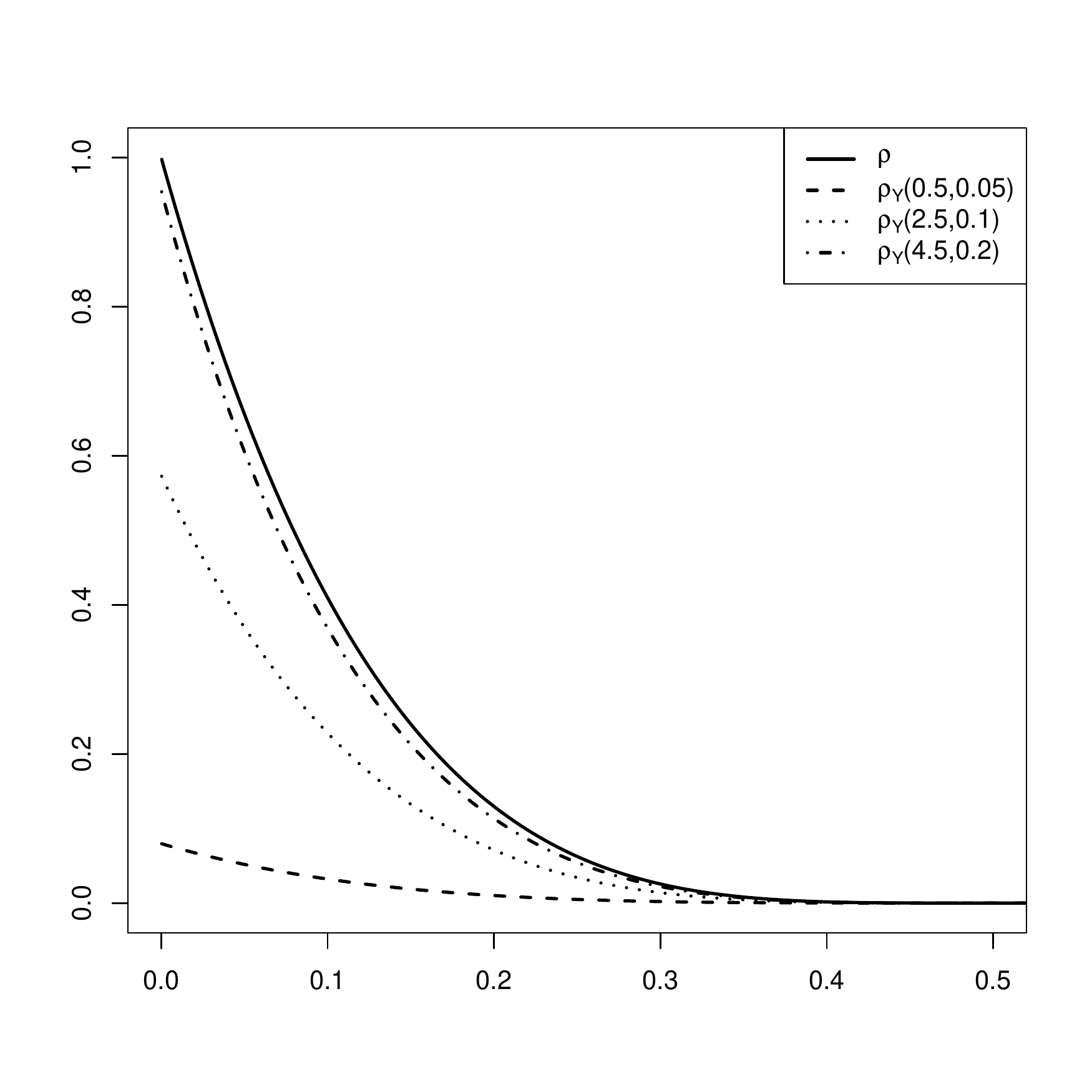}&\includegraphics[width=0.3\textwidth]{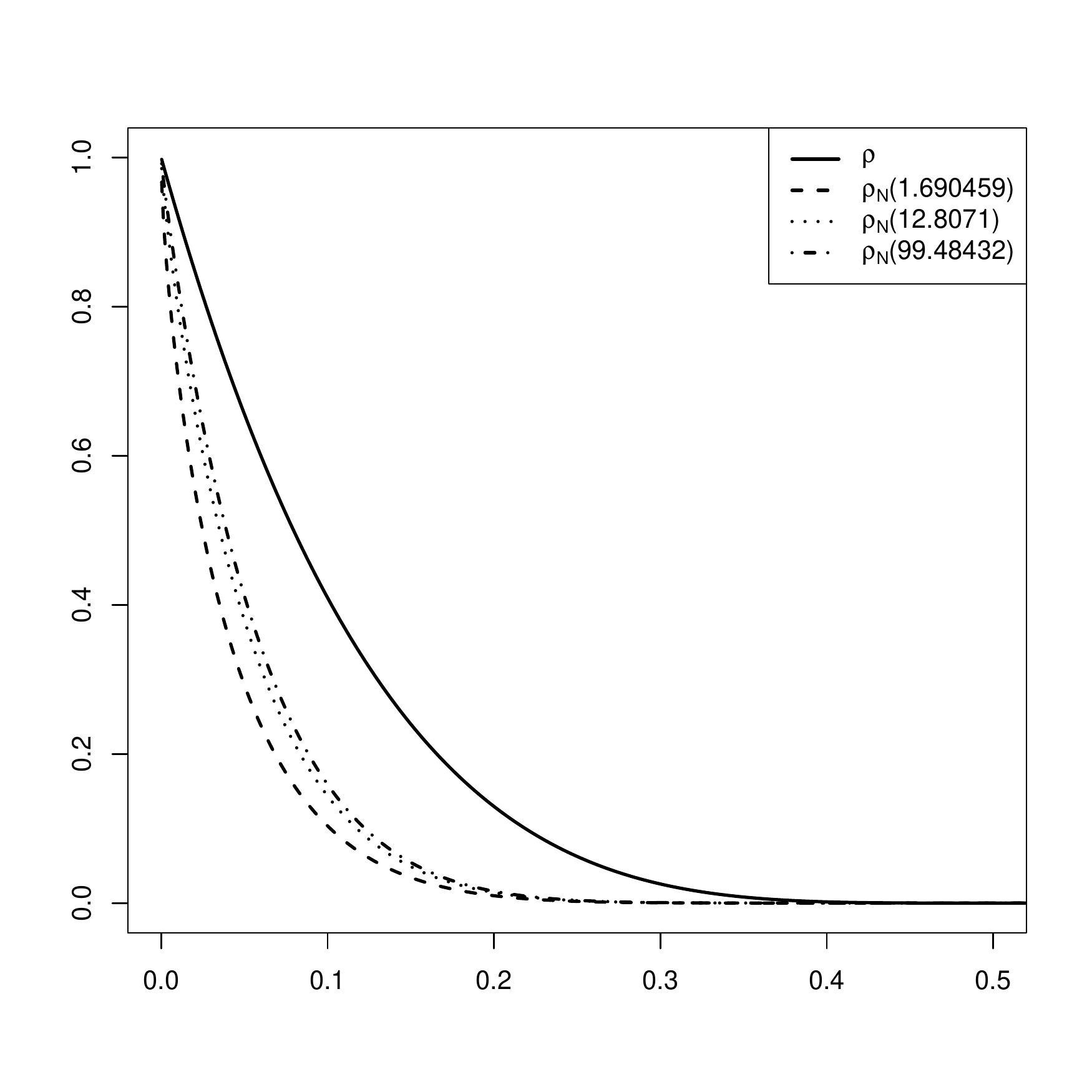}&\includegraphics[width=0.3\textwidth]{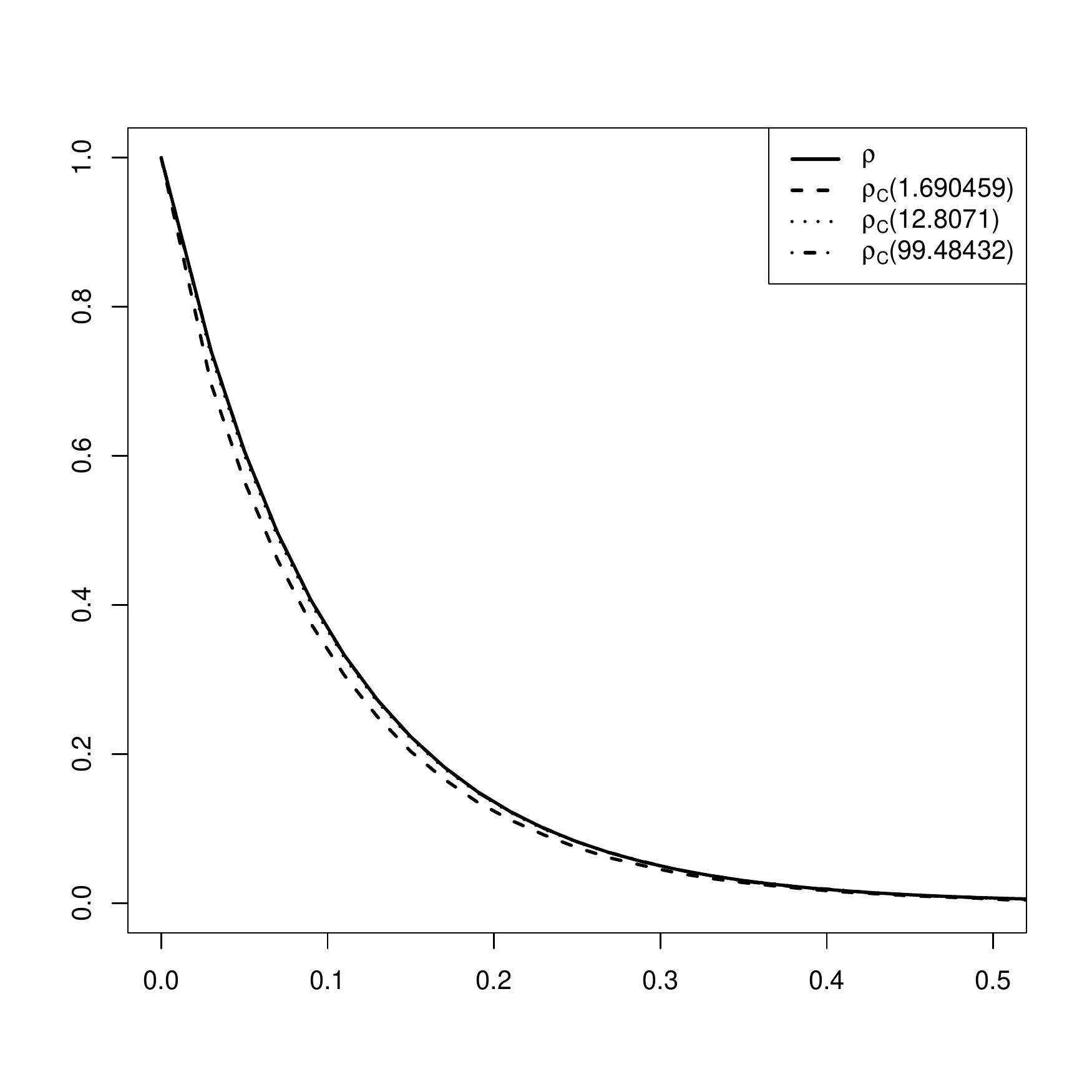}\\
(a)&(b)&(c)\\
\end{tabular}}
\caption{From left to right: 
(a) correlation functions $\rho_Y(\bh,\mu,\sigma^2)$ of the Poisson LG random field with  
$\mu=0.5, \sigma^2=0.05$ and $\mu=2.5, \sigma^2=0.1$, and $\mu=4.5, \sigma^2=0.2$;
(b) correlation function $\rho_N(\bh,\lambda)$ of our proposed Poisson random field for $\lambda=1.69,$ $12.81,$ and, $99.48$;
(c) correlation function $\rho_C(\bh,\lambda)$ of the Poisson GC random field for $\lambda=1.69,$ $12.81,$ and, $99.48$. The black line in the Figures depicts the underlying correlation model given by $\rho(\bh)=(1-||\bh||/0.5)^4)_+$. 
}\label{covaa5}
\end{figure}

\subsection{Bivariate distribution}\label{subsec:3.2}  

In this section, we provide the bivariate distribution of the Poisson random field. This distribution can be written in terms of an infinite series depending on the regularized lower incomplete Gamma function defined in \eqref{incgamma} and the regularized hypergeometric confluent function \citep{Gradshteyn:Ryzhik:2007}, defined as: 
\begin{equation*}
{}_1\widetilde{\mathrm{F}}_{1}(a;b;x)=\frac{{}_1\mathrm{F}_{1}(a;b;x)}{\Gamma(b)}=\sum\limits_{k=0}^{\infty}\dfrac{(a)_k x^k}{\Gamma(b+k)k!},
\end{equation*}
where ${}_1\mathrm{F}_{1}$ is the standard hypergeometric confluent function.

For the sake of simplicity, we analyze the following cases separately: (a) $n=m=0$, (b) $n=0, m\geq1$ and $m=0, n\geq1$, (c) $n=m=1,2 \ldots$,  and (d) $n,m\geq 1$, $n\neq m$. Moreover, we set $p_{nm}=\Pr(N(\bs_i)=n,N(\bs_j)=m)$ , $\lambda_i=\lambda(\bs_i)$,  $\lambda_j=\lambda(\bs_j)$ and $\rho=\rho(\bh)$ for notational convenience. We additionally define the function $\mathcal{S}$ as follows:
\begin{equation*}
\mathcal{S}\left(
        \begin{smallmatrix}
  	a &; b\\
  	& c
	\end{smallmatrix},
        x,x' \right) = {}_1\widetilde{\mathrm{F}}_{1}(a;b;x)\gamma^{\ast}\left(c,x' \right).
\end{equation*}
\begin{theo}\label{theopdf}
Let $N$ be a Poisson random field with underlying correlation $\rho$ and mean $\E(N(\bs_k))=\lambda_k$. Then the bivariate distribution $p_{nm}$ is given by:

\noindent (a) Case $n =m = 0$:
\begin{align*}p_{00}=-1+e^{-\lambda_i}+e^{-\lambda_j}
          +(1-\rho^2)\sum\limits_{k=0}^{\infty}\rho^{2k}\gamma^{\star}\left(k+1,\dfrac{\lambda_i}{1-\rho^2},\dfrac{\lambda_j }{1-\rho^2}\right).
          \end{align*}
\noindent (b) Cases $n \geq 1 , m=0$ and $m \geq 1 , n=0$: 
$$p_{n0}=g(n,\lambda_i,\lambda_j,\rho), \quad p_{0m}=g(m,\lambda_j,\lambda_i,\rho),$$respectively, where
\begin{align*}
g(b,x,y,\rho)=\dfrac{x^b}{b!}e^{-x}
        -x^b e^{-\frac{x}{1-\rho^2}}
        \sum\limits_{\ell=0}^{\infty} \left(\dfrac{\rho^2x}{1-\rho^2}\right)^{\ell}
        \mathcal{S}\left(
        \begin{smallmatrix}
  	b &; b+\ell+1\\
  	& \ell+1
	\end{smallmatrix},
        \dfrac{\rho^2x}{1-\rho^2},\dfrac{y}{1-\rho^2}\right).
        \end{align*}

\noindent (c) Case $n =m \geq 1$:
\begin{align*}
p_{nn}=&-(1-\rho^2)^n\sum\limits_{k=0}^{\infty}\dfrac{\rho^{2k}(n)_k}{k!}\gamma^{\star}\left(n+k,\dfrac{\lambda_i}{1-\rho^2},\dfrac{\lambda_j}{1-\rho^2}\right) \\
       & +\left(\dfrac{1-\rho^2}{\rho^{2}}\right)^n\sum\limits_{k=0}^{\infty}\sum\limits_{\ell=0}^{1}\dfrac{(n)_k}{k!}e^{-\lsi(1-\ell)-\lsj\ell}\gamma^{\star}\left(n+k,\dfrac{\rho^{2(1-\ell)}\lambda_i}{1-\rho^2},\dfrac{\rho^{2\ell}\lambda_j}{1-\rho^2}\right)\\
      &  +(1-\rho^2)^{n+1}\sum\limits_{k=0}^{\infty}\sum\limits_{\ell=0}^{\infty}\dfrac{\rho^{2k+2\ell}(n)_{\ell}}{\ell!}\gamma^{\star}\left(n+\ell+k+1,\dfrac{\lambda_i}{1-\rho^2},\dfrac{\lambda_j}{1-\rho^2}\right).
\end{align*}

\noindent (d) Cases $n \geq 2 , m\geq 1$ with $n>m$, and $m \geq 2 , n\geq 1$ with $m>n$, $$p_{nm}=h(n,m,\lambda_i,\lambda_j,\rho), \quad p_{nm}=h(m,n,\lambda_j,\lambda_i,\rho),$$respectively, where
\begin{align*}
h(a,b,x,y,\rho)=&  x^{m}e^{-\frac{x}{1-\rho^2}}
        \Bigg[\sum\limits_{\ell=0}^{\infty}\dfrac{(b)_\ell}{\ell!}\left(\dfrac{\rho^2x}{1-\rho^2}\right)^{\ell} 	\mathcal{S}\left(
        \begin{smallmatrix}
  	a-b+1 &; a+\ell+1\\
  	& b+\ell
	\end{smallmatrix},
        \dfrac{\rho^2x}{1-\rho^2},\dfrac{y}{1-\rho^2}\right) \\
   &     -\sum\limits_{k=0}^{\infty}\sum\limits_{\ell=0}^{\infty}\dfrac{(b)_\ell}{\ell!}\left(\dfrac{\rho^2x}{1-\rho^2}\right)^{k+\ell}
        \mathcal{S}\left(
        \begin{smallmatrix}
  	a-b &; a+k+\ell+1\\
  	& b+k+\ell+1
	\end{smallmatrix},
        \dfrac{\rho^2x}{1-\rho^2},\dfrac{y}{1-\rho^2}\right)\Bigg].
\end{align*}
\end{theo}

\begin{proof}
See the online supplement.
\end{proof}

The evaluation of the bivariate distribution can be troublesome at first sight. However, it can be performed by truncating the series and considering that efficient numerical computation of the regularized lower incomplete gamma and hypergeometric confluent functions can be found
in different libraries, such as the GNU scientific library \citep{gough2009gnu} and the most important
statistical softwares including \texttt{R}, \textsc{Matlab}, and Python. In particular, the \texttt{R} package Geomodels  \citep{Bevilacqua:2018aa} uses the Python implementations in the SciPy library \citep{2020SciPy-NMeth}.

The bivariate distribution can be written as the product of two independent Poisson distributions when $\rho_N(\bh)=0$. This result, provided by \cite{Hunter:1974} in Theorem 3.6, establishes that the independence of two renewal counting processes is equivalent to a zero correlation between them. As outlined in Section \ref{sec:3.1}, $\rho(\bh)=0$ implies $\rho_N(\bh)=0$. Consequently, pairwise independence at the level of the underlying Gaussian random field implies pairwise independence for the Poisson random field.

We now compare the type of bivariate dependence induced by the proposed model and the GC one
when $\lambda=5$. Figure (\ref{ddddd}) (from left to right) presents the bivariate GC distribution, the bivariate Poisson distribution in Theorem \ref{theopdf} and a coloured image representing the differences between them. Note that a positive value of the difference implies that probabilities associated with bivariate distribution in Theorem \ref{theopdf} are greater than the probabilities of the bivariate GC one. Only the probabilities $\Pr(N(\bs_i)=n,N(\bs_j)=m)$ for $n,m=0,1,\ldots, 12$ are considered in the plots. The first, second and third row consider increasing levels of underlying correlations $\rho(\bh)=0.1, 0.5, 0.9$. 
It can be appreciated that the larger the correlation, the more significant is the difference between the proposed and Gaussian copula bivariate distributions. In addition, there is a pattern in which the probabilities of the GC bivariate distribution tend to be larger along the diagonal, {\it i.e.}, the blue scale color is predominant along the diagonal.
This is not surprising since the Poisson GC bivariate distribution inherits the type of dependence of the bivariate Gaussian distribution.


 \begin{figure} [ht!]
 \vspace{-0.5cm}
\scalebox{0.9}{
\begin{tabular}{ccc}
\includegraphics[width=0.3\textwidth]{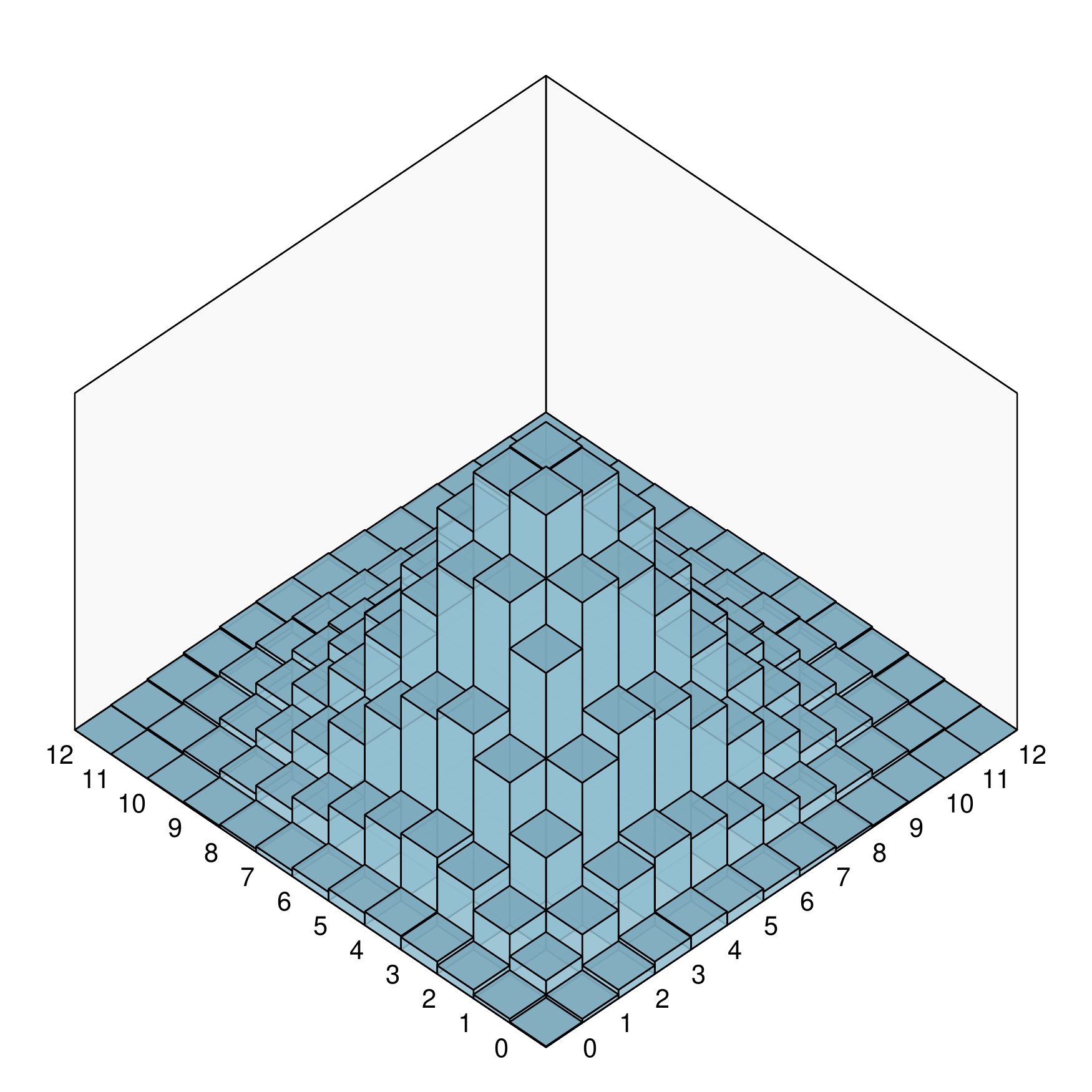}&\includegraphics[width=0.3\textwidth]{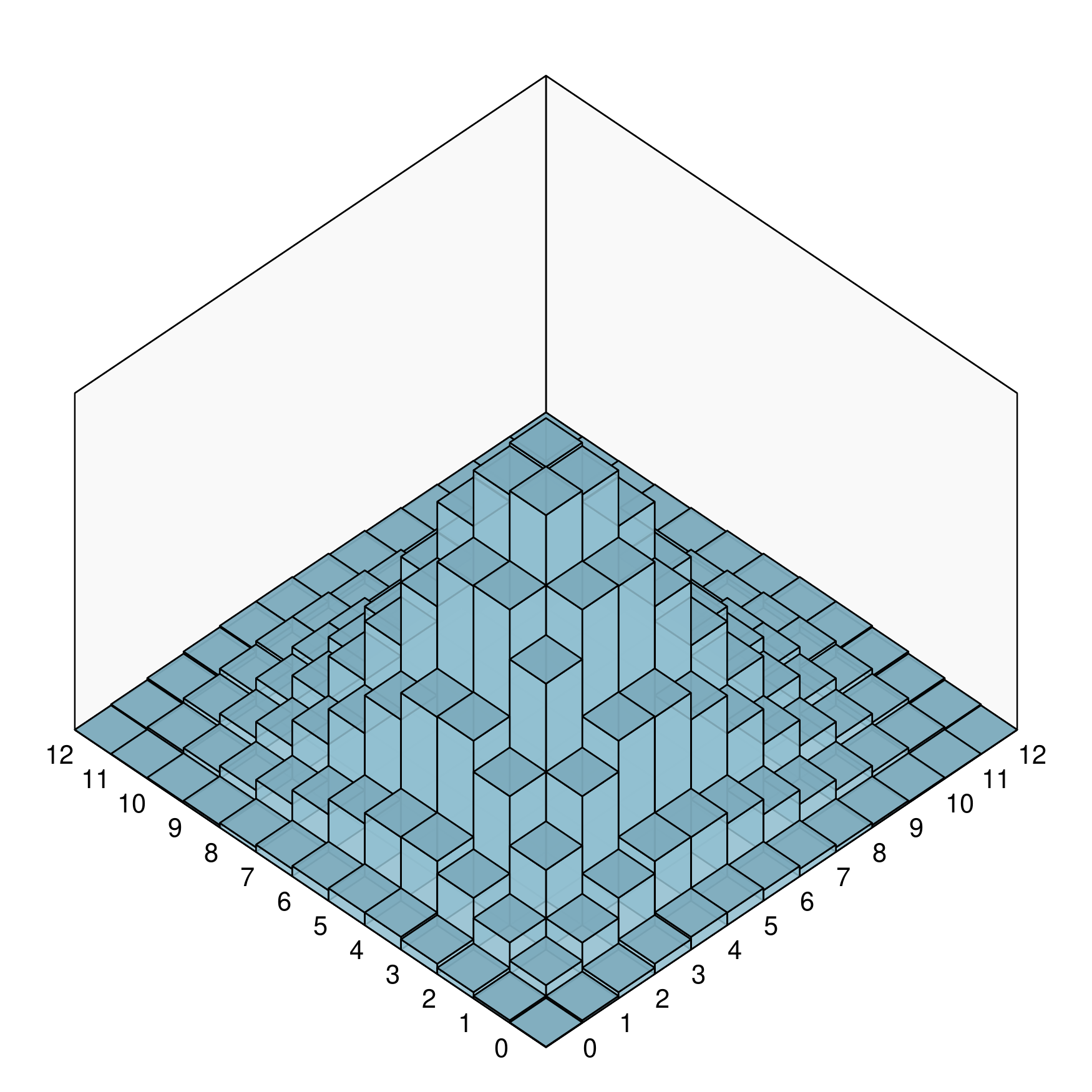}&\includegraphics[width=0.33\textwidth]{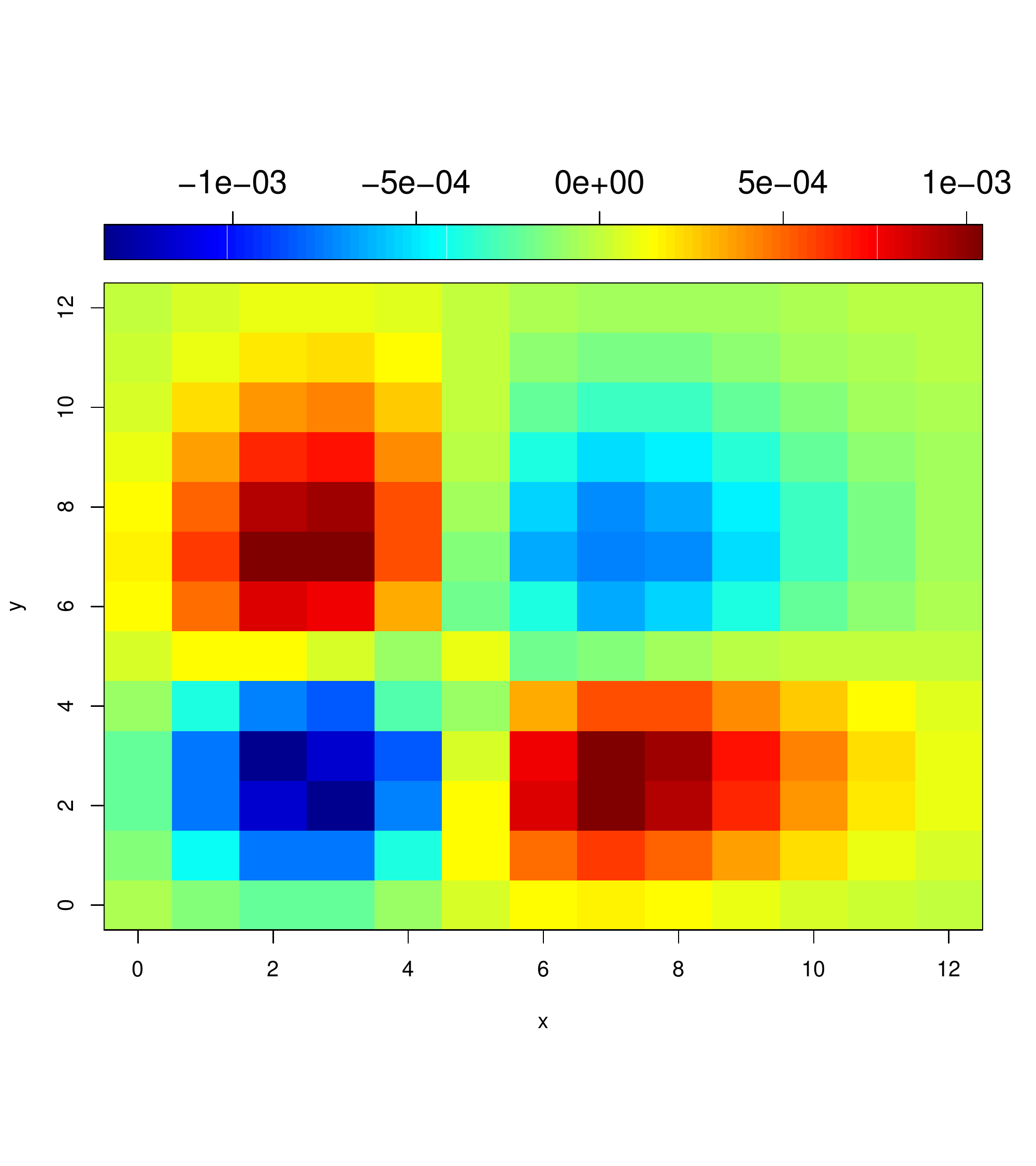}\\ \vspace{-0.3cm}
\includegraphics[width=0.3\textwidth]{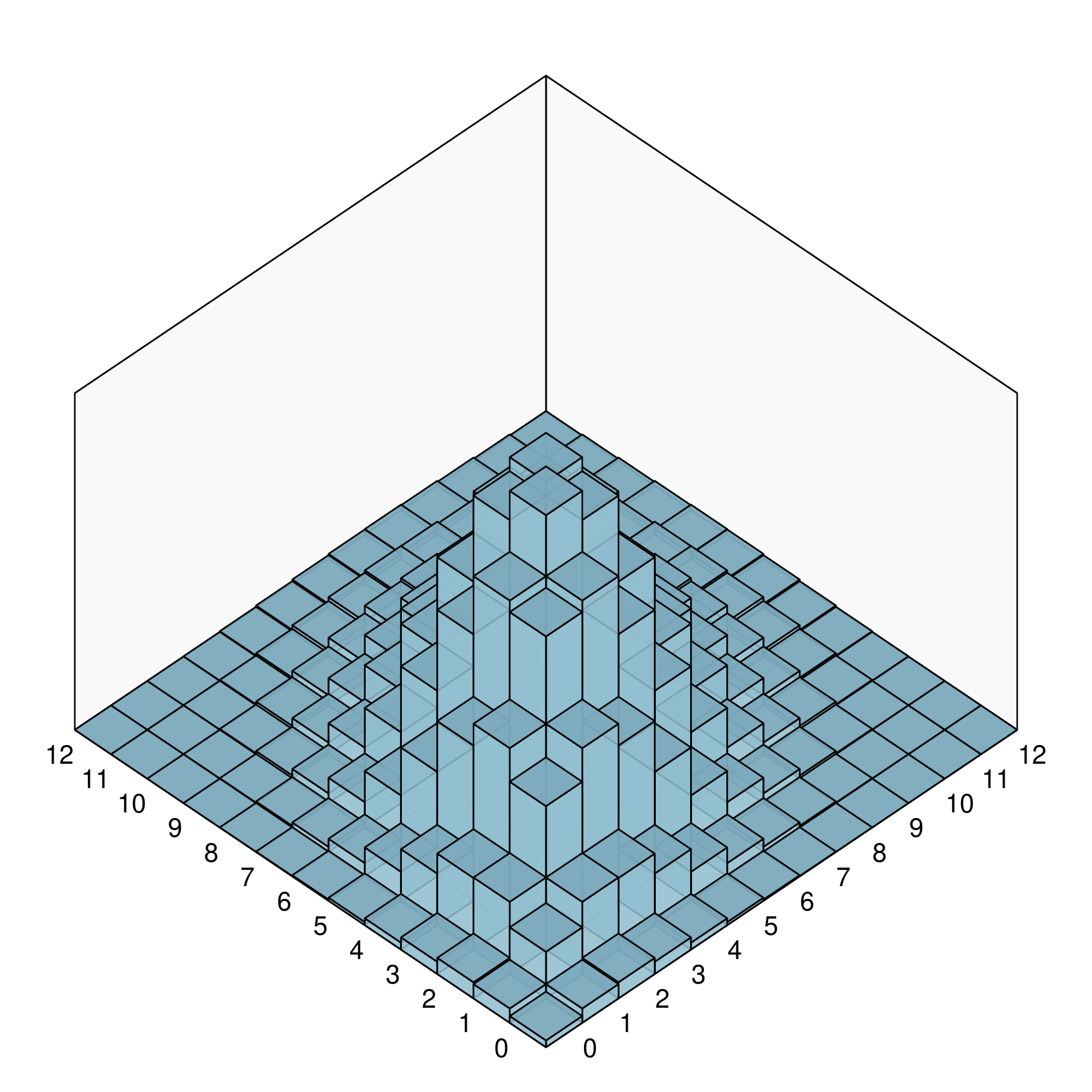}&\includegraphics[width=0.3\textwidth]{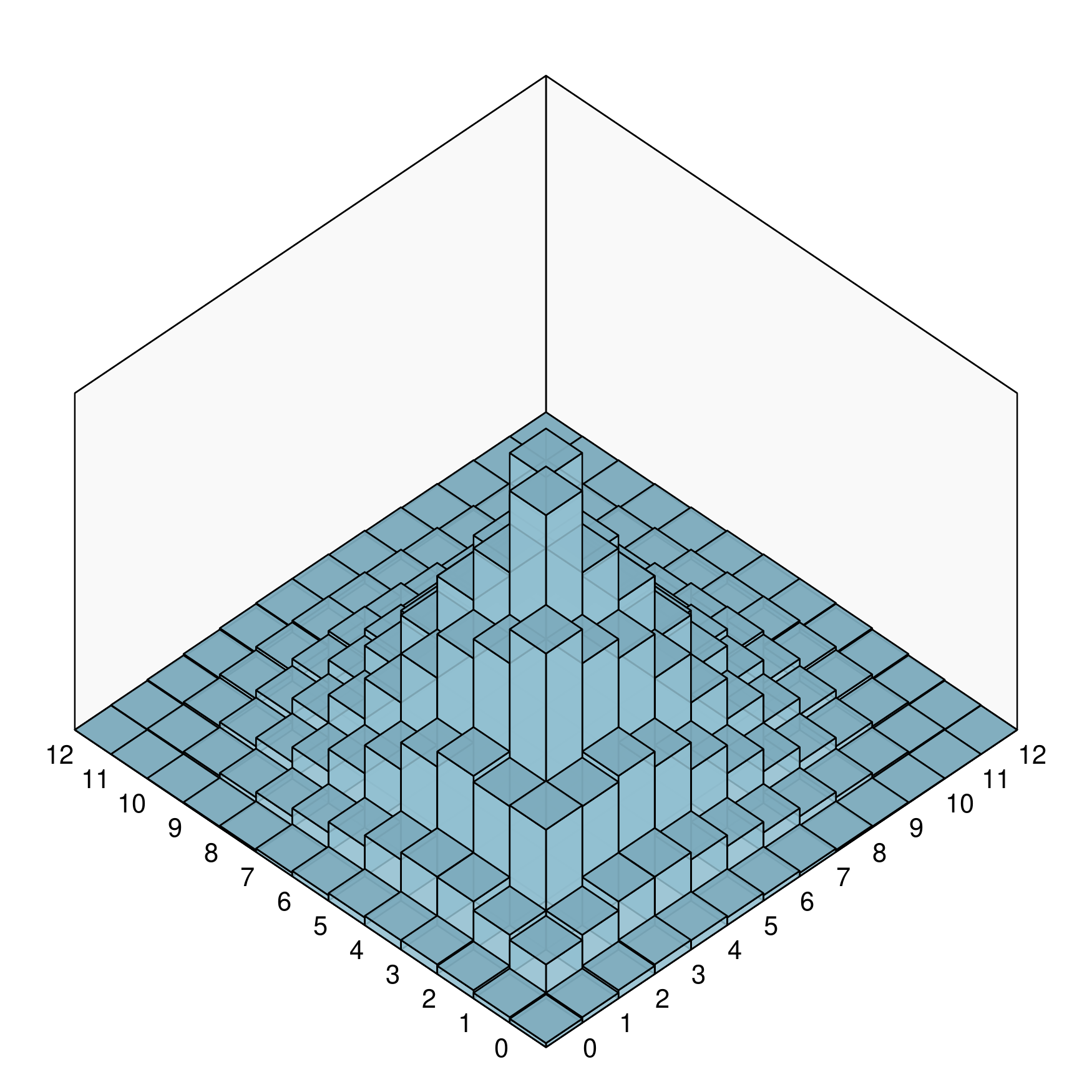}&\includegraphics[width=0.33\textwidth]{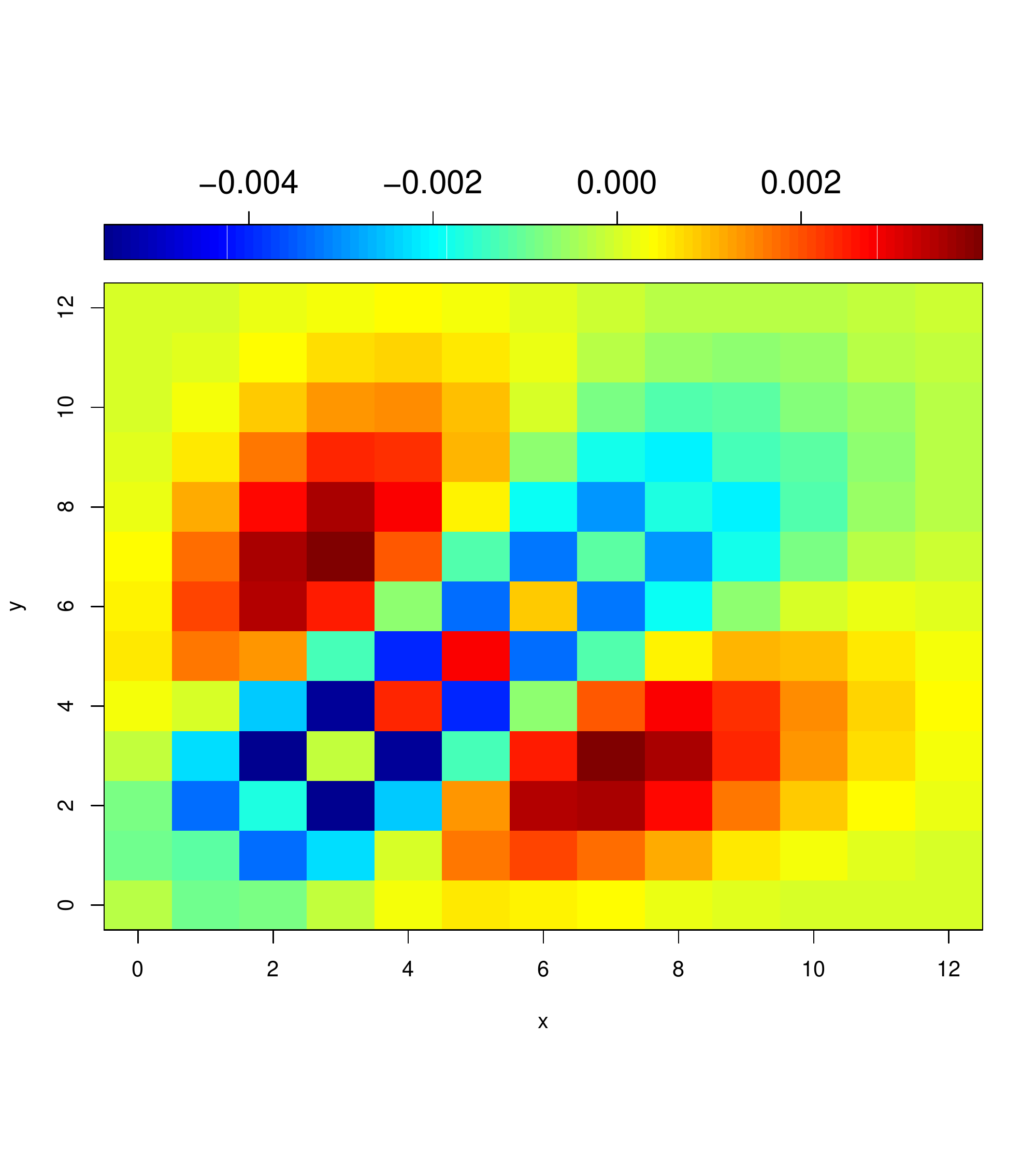}\\ \vspace{-0.3cm}
\includegraphics[width=0.3\textwidth]{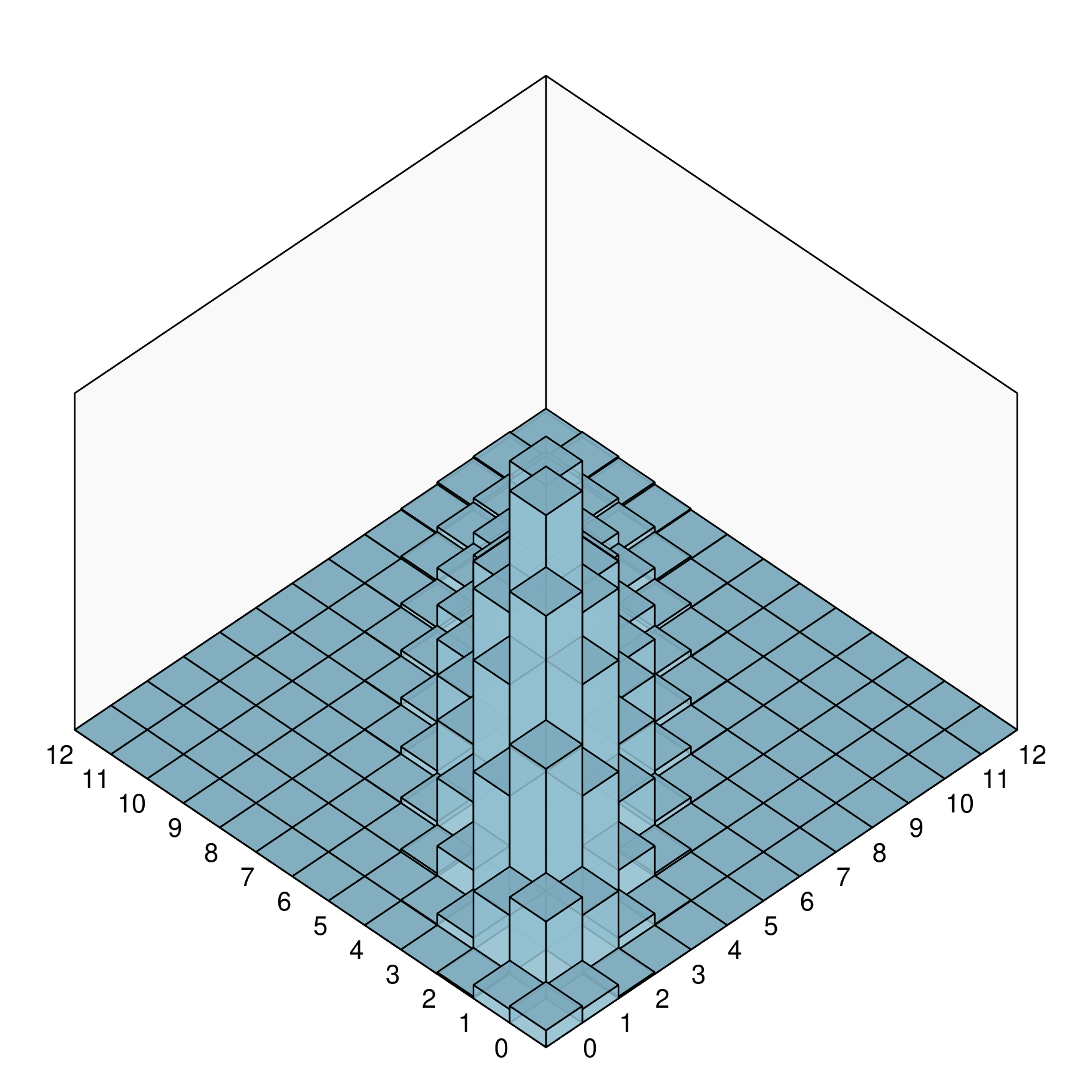}&\includegraphics[width=0.3\textwidth]{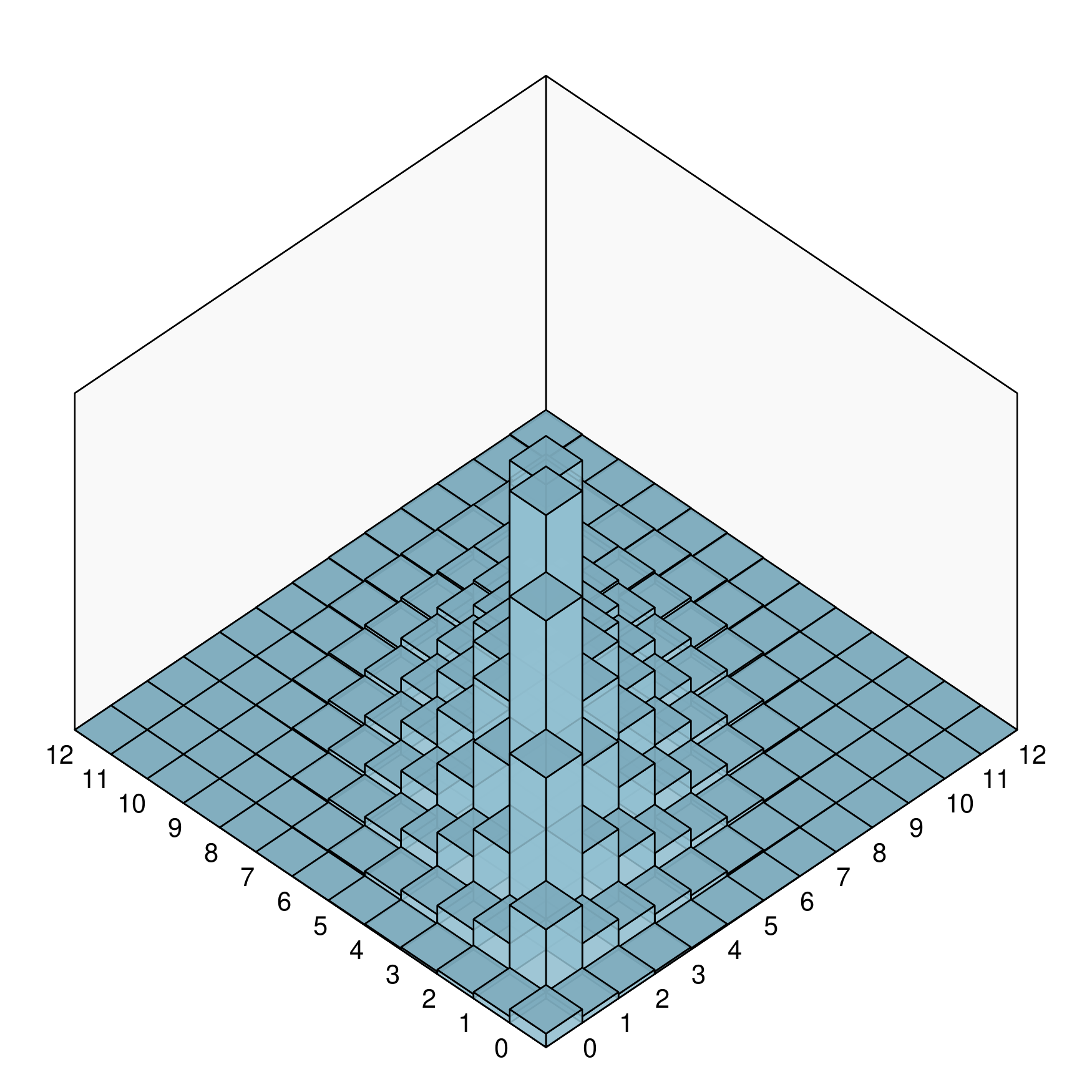}&\includegraphics[width=0.33\textwidth]{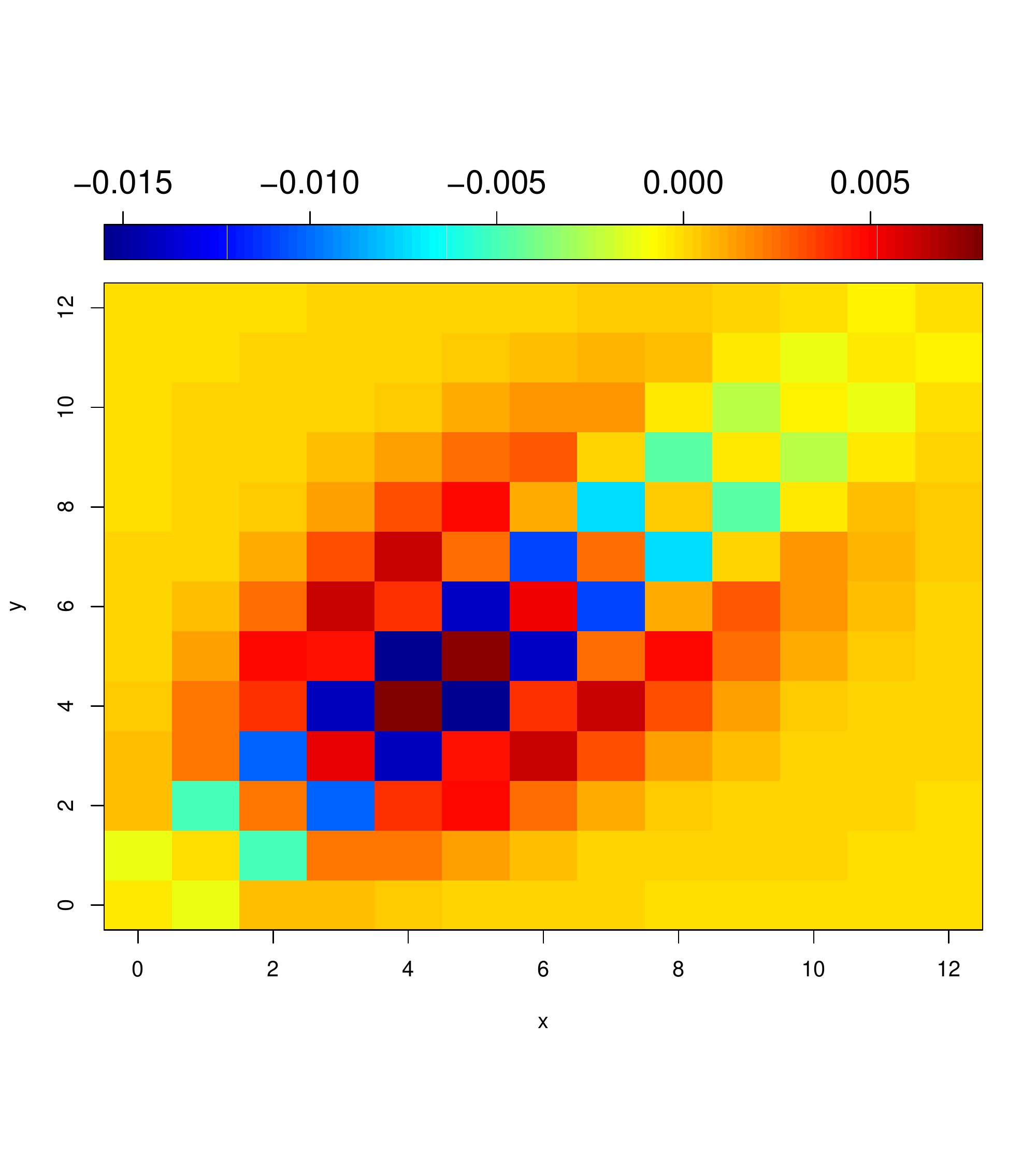}\\
(a)&(b)&(c)\\
\end{tabular}}
\caption{For each row (from left to right): bivariate Poisson GC distribution, our proposed bivariate Poisson distribution and the difference between them. The first, second and third row are obtained setting $\rho(\bh)=0.1,0.5, 0.9$ for the underlying correlation.}\label{ddddd}
\end{figure}

\subsection{A zero-inflated extension}\label{subsec:3.3} 
In this Section we provide an extension of the proposed model
for spatial data which exhibit an excessive number of zeros.
Specifically, let $B = \{B(\bs),\bs\in A\}$, a Bernoulli random field such that $B(\bs)=\mathds{1}_{(-\infty,0)}(G(\bs))$, where $G$ is a Gaussian random field with $\mathbb{E}(G(\bs)) = \theta(\bs)$, unit variance and correlation function $\rho_1(\bh)$. 
The marginal probability of having a zero is then given by:
\vspace{-0.5cm}
\[p(\bs):=\Pr(B(\bs)=0)=\Phi(\theta(\bs)),\vspace{-0.5cm}\]
where $\Phi$ is the univariate standard Gaussian {\it cdf}.   
Let $N$ be a Poisson random field with $\mathbb{E}(N(\bs)) = \lambda(\bs)$ and underlying correlation $\rho_2(\bh)$. Assuming $B$ and $N$ are independent, the proposed Poisson zero-inflated model is then given by a random field $Y= \{Y(\bs), \bs\in A\}$ defined as:
\vspace{-0.5cm}
\begin{equation}\label{ppppp}
Y(\bs):=B(\bs)N(\bs),
\vspace{-0.5cm}
\end{equation}	
with marginal distribution given by:  
\begin{equation}\label{mzip}
\Pr(Y(\bs)=y(\bs))=
	\begin{cases}p(\bs)+(1-p(\bs))e^{-\lambda(\bs)}
	&\mbox{if} \quad y(\bs)=0\\
(1-p(\bs))\dfrac{\lambda(\bs)^{y(\bs)}e^{-\lambda(\bs)}}{y(\bs)!} & \mbox{if} \quad    y(\bs)=1,2,\ldots \end{cases},
\end{equation}	

and with $\E(Y(\bs))=(1-p(\bs))\lambda(\bs)$ and $\Var(Y(\bs))=\E(Y(\bs))[1+\frac{p(\bs)}{1-p(\bs)}\E(Y(\bs))]$. 
Note that the zero-inflated Poisson random field is overdispersed and when $p(\bs)\to 0$ then the Poisson random field is obtained as special case. For the sake of completeness, we provide the bivariate distribution and the correlation function of the zero-inflated Poisson random field in Proposition 1 (see the online supplement).

\section{Estimation and prediction}\label{sec:4}

In this section, we start by describing the weighted pairwise likelihood ({\it wpl}) estimation method; then, we focus on the optimal linear prediction. 

\subsection{Weighted pairwise likelihood estimation}\label{opop}

Composite likelihood is a general class of objective functions that combine low-dimensional terms based on the likelihood of marginal or conditional events to construct a pseudo likelihood \cite{Lindsay:1988,Varin:Reid:Firth:2011}. A particular case of the composite likelihood class is the pairwise likelihood \citep[see for example][for application of pairwise likelihood  in the spatial setting] {Heagerty:Lele:1998,Bevilacqua:Gaetan:2015,ALEGRIA2017,Bev:2020} that combines the bivariate distributions of all possible distinct pairs of
observations. Let $\bm{N}=(n_1,n_2,\ldots,n_l)^{\top}$ be a realization of the Poisson random field $N$ observed at distinct spatial locations $\bs_1,\bs_2,\ldots,\bs_l$, $\bs_i\in A$ and let $\bm{\theta}=(\bm{\beta}^{\top},\bm{\alpha}^{\top})$ be the vector of unknown parameters where $\bm{\alpha}$ is the vector parameter associated with the underlying correlation model and $\bm{\beta}$ the regression parameters. The pairwise likelihood function is defined as follows:
\begin{equation*}\label{ppl}
\mathrm{pl}(\bm{\theta}):= \sum\limits_{i=1}^{l-1}\sum\limits_{j=i+1}^{l}\log(\Pr(N(\bs_i)=n_i,N(\bs_j)=n_j))\zeta_{ij},
\end{equation*}  
where $\Pr(N(\bs_i)=n_i,N(\bs_j)=n_j)$ is the bivariate density given in Theorem \ref{theopdf} and $\zeta_{ij}$ is a non-negative suitable
weight. The choice of cut-off weights, namely,
\begin{equation}\label{wer}
\zeta_{ij}= 
	\begin{cases}
	1 	&\parallel\bs_i-\bs_j \parallel\leq \xi \\
	0 & \text{otherwise}  \end{cases},
\end{equation}
for a positive value of $\xi$, can be motivated by its simplicity and by observing that the dependence between distant observations is weak \citep{Joe:Lee:2009}. Some guidelines on the choice of $\xi$
can be found in \cite{Bevilacqua:Gaetan:Mateu:Porcu:2012,Bevilacqua:Gaetan:2015}.

The maximum weighted pairwise likelihood ({\it wpl}) estimator is given by:
\vspace{-0.5cm}
\begin{equation*}
\widehat{\bm{\theta}}:=\operatorname{argmax}_{\bm{\theta}}\, \operatorname{pl}(\bm{\theta}).
\vspace{-0.5cm}
\end{equation*}
Under some mixing conditions of the Poisson random field, following \cite{Bevilacqua:Gaetan:2015}, it can be shown that, when increasing domain asymptotic, $\widehat{\bm{\theta}}$ is consistent and asymptotically Gaussian distributed, with the  covariance matrix given by $\mathcal{G}^{-1}_n(\bm{\theta})$, {\it i.e.,} the inverse of the Godambe information
$\mathcal{G}_n(\bm{\theta}):=\mathcal{H}_n(\bm{\theta})\mathcal{J}_n(\bm{\theta})^{-1}\mathcal{H}_n(\bm{\theta}),
$
where
$\mathcal{H}_n(\bm{\theta}):=\E[-\nabla^2 \operatorname{pl}(\bm{\theta})]$ and $\mathcal{J}_n(\bm{\theta}):={\mbox{Var}}[\nabla \operatorname{pl}(\bm{\theta})]$. The standard error estimation can be obtained from the square root diagonal elements of $\mathcal{G}^{-1}_n(\widehat{\bm{\theta}})$.

It is important to stress that the computation of the standard errors requires the evaluation of the matrices $\mathcal{H}_n(\hat{\bm{\theta}})$ and $\mathcal{J}_n(\hat{\bm{\theta}})$. However, the evaluation of $\mathcal{J}_n(\hat{\bm{\theta}})$ is computationally
unfeasible for large datasets, and in this case, sub-sampling techniques can be used, as in \cite{Heagerty:Lele:1998} and \cite{Bevilacqua:Gaetan:Mateu:Porcu:2012}. A straightforward and more robust alternative is the parametric bootstrap estimation of $\mathcal{G}^{-1}_n(\bm{\theta})$ \citep{Bai:Kang:Song:2014}.

Another critical issue related to large datasets is that computation of the {\it wpl} estimator can be computationally demanding 
due to the computational complexity associated with the bivariate Poisson distribution given in Theorem \ref{theopdf}.
An estimator that requires a smaller computational burden can be obtained under
Gaussian misspecification \citep{AOS1121,cppp,Bev:2020}. This is a useful inferential tool when the likelihood function cannot be calculated for some reason, but the first two moments and the correlation are known. 

In our case, we assume a non-stationary Gaussian random field with mean and variance equal to $\lambda(\bs)$ and correlation $\rho_{N}(\bs_i,\bs_j)$ given in Theorem \ref{Theo1}. 
Then the  misspecified maximum $wpl$ requires the computation of the Gaussian bivariate distribution and  the misspecified standard maximum likelihood estimation 
requires the computation of the Gaussian multivariate distribution.
In both cases, evaluation of the Gamma incomplete function or the modified Bessel function (in the stationary case) is required to compute the covariance matrix. Table 1, in the online supplement, shows the computational cost of each estimation procedure through different scenarios for the stationary case, to demonstrate the computational gains of the Gaussian misspecified estimation with respect to the Poisson {\it wpl} estimation.

\subsection{Optimal linear prediction}\label{olpred}

The Poisson random field's optimal predictor concerning the mean squared error criterion requires the knowledge of the finite-dimensional distribution, which is not available for the Poisson random field. As in the estimation step, once again, the Gaussian misspecification allows to build 
the best linear unbiased predictor (BLUP) based on the correlation of the Poisson random field given in Theorem \ref{Theo1}. Specifically, if the goal is the prediction of $N$ at $\bs_0$ given the vector of spatial observations $\bN$ observed at $\bs_1,\bs_2,\ldots,\bs_l$, then the optimal linear Gaussian prediction is given by: 
\begin{equation}\label{pyt}
\widehat{N(\bs_0)}=\lambda(\bs_0)+ \bm{c}^\top \Sigma^{-1}(\bm{N}-\bm{\lambda}),
\end{equation}
where $\bm{\lambda}=(\lambda(\bs_1),\ldots,\lambda(\bs_l))^\top$, $\bm{c}=[\sqrt{\lambda(\bs_0)\lambda(\bs_i)}\rho_{N}(\bs_0,\bs_i)]_{i=1}^l$ and
$\Sigma=\sqrt{\bm{\lambda} \bm{\lambda}^\top} \odot [\rho_{N}(\bs_i,\bs_j)]_{i,j=1}^l$ 
is the variance-covariance matrix ($\odot$ the matrix Schur product). In practice, the mean and covariance matrix are not known and must be estimated. The associated mean squared error is:
\vspace{-0.5cm}
\begin{equation*}
\mbox{MSE}(\widehat{n(\bs_0)})=\lambda(\bs_0)  - \bm{c}^\top\Sigma^{-1}\bm{c}.
\end{equation*}

Note that this kind of prediction does not guarantee the positivity and discreteness of the prediction. However, optimal linear prediction can generally be a useful approximation of the optimal predictor, as was shown, for example, in \cite{DeOliveira:2006} and recently in \cite{Bevilacqua:2018ab}.
When $l$ is large the use of compactly supported correlation functions \citep{Bevilacqua:Faouzi_et_all:2019} can mitigate the computational burden associated with the optimal linear predictor
since sparse matrix algorithms can be exploited to handle the inverse of the covariance matrix efficiently.

\section{Simulation studies}\label{sec:5}

In this section, we focus on two simulation studies. The first one analyses the performance of the {\it wpl} method when estimating the Poisson random field under the spatial and spatio-temporal settings. The second one analyses the Poisson optimal linear predictor's performance, comparing our approach with the Poisson GC and Poisson LG models. This study is presented in the online supplement.

\subsection{Performance of the {\it wpl} estimation}\label{pol}

In this study, we consider $1000$ realizations from a stationary spatial Poisson random field observed at $\bs_i\in [0,1]^2$, $i=1,\ldots , l$, $l = 441$. Specifically, we considered a regular grid with increments of size $0.05$ over the unit square $[0, 1]^2$. The grid points were perturbed, adding a uniform random value over $[-0.015, 0.015]$ to each coordinate. A perturbed grid allows us to obtain more stable estimates since different sets of small distances are available and very close location points are avoided. 
The simulation of the Poisson random field follows directly from the stochastic representation in \eqref{qqq}, and it depends on the simulation of a sequence of independent copies of an exponential random field obtained by transforming independent copies of a standard Gaussian random field. These random fields were simulated using the Cholesky decomposition.

For the Poisson random field we, consider $\lambda(\bs)=e^{\beta}$ with 
 $\beta=\log(2),$ $\log(5),$ $\log(10),$ $\log(20)$, and an underlying isotropic correlation model $\rho(\bh)=(1-||\bh||/\alpha)^{4}_+$ with $\alpha=0.2$. As outlined in Section \ref{subsec:3.2}, the use of a compactly supported correlation function simplifies the computation of the bivariate Poisson distribution proposed in Theorem \ref{theopdf}. 
 
We study the performance of the Poisson {\it wpl}, the misspecified Gaussian {\it wpl} and the misspecified Gaussian maximum likelihood (ML) estimation methods. In the (misspecified) {\it wpl} estimation, we consider a cut-off weight function, as in \eqref{wer}, with $\xi = 0.1$.

\begin{table}[htb!]
\centering
\scalebox{0.80}{
\begin{tabular}{|c|cc|cc|cc|}
  \cline{2-7} 
\multicolumn{1}{c|}{} &\multicolumn{2}{c|}{Poisson {\it wpl}} & \multicolumn{2}{c|}{Gaussian {\it wpl}}&\multicolumn{2}{c|}{Gaussian ML}\\
  \cline{2-7} 
 \multicolumn{1}{c|}{} & Bias &MSE&  Bias & MSE &  Bias & MSE  \\ 
  \hline
 $\beta=\log(2)$  &-0.00251 & 0.00151 &  -0.00348  &0.00161  & -0.00366  & 0.00165  \\ 
 $\alpha=0.2$ & -0.00663   & 0.00113 &-0.00828 & 0.00208&-0.00748    & 0.00203 \\ 
   \hline
 $\beta=\log(5)$ &-0.00113 &0.00065&  -0.00147 & 0.00068& -0.00161 &  0.00068 \\  
 $\alpha=0.2$ & -0.00435&  0.00098&  -0.00422 & 0.00149 & -0.00344 & 0.00145  \\ 
    \hline
    $\beta=\log(10)$ &    0.00052 & 0.00033 &0.00031&  0.00033  &0.00014 & 0.00033  \\ 
 $\alpha=0.2$   &  -0.00261  &0.00096 &-0.00336 & 0.00120& -0.00296 & 0.00115  \\ 
   \hline
    $\beta=\log(20)$ & -0.00026 &    0.00018 &-0.00039   &  0.00019 &-0.00037  &   0.00018 \\ 
 $\alpha=0.2$   &  -0.00449&  0.00094& -0.00499  &0.00099& -0.00402 & 0.00095  \\ 
   \hline
\end{tabular}}
\caption{Bias and MSE associated with Poisson {\it wpl}, misspecified Gaussian {\it wpl} and misspecified Gaussian ML when the true random field is Poisson with $\lambda(\bs)=e^{\beta}$ and $\rho(\bh)=(1-||\bh||/\alpha)^{4}_+$.}\label{simu11}
\end{table}

Table \ref{simu11} shows the bias and mean squared error associated with $\beta$ and $\alpha$ through the four scenarios and three estimation methods. As expected, the misspecified Gaussian ML performs slightly better than the misspecified Gaussian {\it wpl}. More importantly, it can be recognized that the Poisson {\it wpl} shows the best performance, particularly when estimating the spatial dependence parameter. This fact is more evident for low counts {\it i.e.}, when $\beta$ is decreasing. However, when increasing the mean, the performances of the three methods of estimation tend to be considerably similar, in particular when the mean of the Poisson random field is $20$.

To summarize, the Poisson {\it wpl} is the best method for estimating the Poisson random field when the mean is small (lower than 20 as a rule of thumb in our experiments). For large counts, the misspecified Gaussian {\it wpl} or ML methods show approximately the same performance as the Poisson {\it wpl} method.
 
We also study the proposed methods' performance when estimating a non-stationary version of the Poisson random field. Under the previous simulation setting we changed the constant mean by considering a regression model, that is, $\lambda(\bs)=\exp\{\beta+\beta_1u_1(\bs)+\beta_2 u_2(\bs)\}$ with $\beta=1.5$, $\beta_1=-0.2$ and $\beta_2=0.3$, where $u_1(\bs)$ and $u_2(\bs)$ are independent realizations from a $(0,1)$ uniform random variable. Table \ref{simu22} shows the bias and MSE associated with $\beta$, $\beta_1$, $\beta_2$ and $\alpha$ for the three methods of estimation, and Figure 1 of the online supplement plots the associated centred boxplots. 

Additionally, in this case, the Poisson {\it wpl} method shows the best performance. We replicate the simulation by considering larger values of the regression parameters (the results are not reported here), which lead to larger counts, and the three methods of estimations show approximately the same performance as in the stationary case.

\begin{table}[htb!]
\centering
\scalebox{0.80}{
\begin{tabular}{|l|cc|cc|cc|}
  \cline{2-7}
\multicolumn{1}{c|}{} &\multicolumn{2}{c|}{Poisson {\it wpl}}& \multicolumn{2}{c|}{Gaussian {\it wpl}}&\multicolumn{2}{c|}{Gaussian ML}\\
  \cline{2-7}
\multicolumn{1}{c|}{} & Bias &MSE&  Bias & MSE &  Bias & MSE  \\ 
  \hline
 $\beta=1.5$ &  -0.00263 & 0.00419  &-0.00359  &0.00445  & -0.00320&  0.00428  \\ 
  $\beta_1=-0.2$ &  0.00189  &0.00618  &0.00148 & 0.00665 &0.00046 & 0.00627  \\ 
    $\beta_2=0.3$ &   0.00185& 0.00608& 0.00202 &0.00661 &0.00182 &0.00621\\ 
 $\alpha=0.2$ &   -0.00148 &0.00091&   -0.00030& 0.00124& 0.00096& 0.00122 \\ 
   \hline
\end{tabular}}\caption{Bias and MSE associated
with the Poisson {\it wpl}, misspecified Gaussian {\it wpl} and misspecified Gaussian ML when estimating a non-stationary Poisson random field with $\lambda(\bs)=\exp\{\beta+\beta_1u_1(\bs)+\beta_2 u_2(\bs)\}$ and  $\rho(\bh)=(1-||\bh||/\alpha)^{4}_+$.} \label{simu22}
\end{table}

Finally, we consider a simulation scheme under a spatio-temporal setting. Specifically, we consider $1000$ simulations from a non-stationary space-time Poisson random field observed at $\bs_i\in [0,1]^2$, $i=1,\ldots, l$, $l = 40$ spatial location sites, uniformly distributed within the unit square and $t^*_1=0, t^*_2=0.25, \ldots t^*_{25} =6$, $25$ time points. We consider a regression model for the spatio-temporal mean $\lambda(\bs,t^*)=\exp\{\beta+\beta_1u_1(\bs,t^*)+\beta_2 u_2(\bs,t^*)\}$, where $u_k(\bs,t^*)$, $k=1,2$ are independent realizations from a $(0,1)$ uniform random variable. We set $\beta=1.5$, $\beta_1=-0.2$ and $\beta_2=0.3$ as in the previous simulation scheme.

Additionally, as the underlying space-time correlation, we use a simple isotropic and temporal symmetric space-time Wendland separable model $\rho(\bh,t^\star)=(1-||\bh||/\alpha_\bs)^{4}_+ (1-|t^\star|/\alpha_{t^*})^{4}_+$
with $\alpha_\bs=0.2$ and $\alpha_{t^*}=1$, where $t^\star=t^*_i-t^*_j$ with $i,j\in \{1,2,\ldots,25\}$. Finally, for the misspecified \textit{wpl} estimation, we consider a cut-off weight function as in \eqref{wer} extended to the space time case, with $\xi_{\bs} = 0.2$ and $\xi_{t^*} =0.5$. 

The results concerning this simulation study are shown in Table \ref{simu22st}, including the bias and MSE associated with $\beta$, $\beta_1$, $\beta_2$ and $\alpha_\bs$, $\alpha_{t^*}$ for the three estimation methods. In addition, Figure 2 of the online supplement shows the associated box plots. As it can be observed, the Poisson {\it wpl} approach outperforms the misspecified Gaussian {\it wpl} and ML as expected for each parameter.

We want to highlight that we have replicated this simulation study by considering higher regression parameters' values, leading to larger counts (for space reasons, these results are not reported here). In that case, all of the estimation methods showed similar behaviours as in the purely spatial case.

 \begin{table}[htb!]
\centering
\scalebox{0.8}{
\begin{tabular}{|l|cc|cc|cc|}
  \cline{2-7}
\multicolumn{1}{c|}{} &\multicolumn{2}{c|}{Poisson {\it wpl}}& \multicolumn{2}{c|}{Gaussian {\it wpl}}&\multicolumn{2}{c|}{Gaussian ML}\\
  \cline{2-7}
\multicolumn{1}{c|}{} & Bias &MSE&  Bias & MSE &  Bias & MSE  \\ 
  \hline
 $\beta=1.5$ &    -0.00058 & 0.00166&  -0.00120 & 0.00180 & -0.00110 & 0.00167    \\ 
  $\beta_1=-0.2$ &  -0.00079 & 0.00257& -0.00056 &  0.00274 & -0.00102  &0.00249  \\
    $\beta_2=0.3$ &0.00036   & 0.00284  &0.00070  & 0.00302 &0.00062 &  0.00267    \\
 $\alpha_\bs=0.2$ &      -0.01057 & 0.00464 &-0.01323 & 0.00630 &-0.01343& 0.00629    \\
  $\alpha_{t^*}=1$ &    -0.00124 & 0.01846& 0.00165 & 0.02534 &0.00032 & 0.02415\\ 
   \hline
\end{tabular}}\caption{Bias and MSE associated
with the Poisson {\it wpl}, misspecified Gaussian {\it wpl} and misspecified  Gaussian ML when estimating a non-stationary spatiotemporal Poisson random field with $\lambda(\bs,t^*)=\exp\{\beta+\beta_1u_1(\bs,t^*)+\beta_2 u_2(\bs,t^*)\}$ and  $\rho(\bh,t^\star)=(1-||\bh||/\alpha_\bs)^{4}_+(1-|t^\star|/\alpha_{t^*})^{4}_+$.} \label{simu22st}
\end{table}


\section{Application to the  reindeer pellet-group survey in Sweden}\label{sec:6}

As mentioned in the Introduction, the pellet-group survey is a technique that provides a general idea of species distribution over a specific geographic area. This technique is used, mainly, for (a) estimating the population density of several ungulate species such as deer \citep[see for example][among many others]{EberhardtVanEtten1956,FreedyBowden1983,MootyEtAl1984,RowlandEtAl1984,MarquesEtAl2001} and (b) to study the impact of some covariates in the election of their habitat selection \citep[see for example][]{SkarinEtAl2015,Lee2016,SkarinEtAl2017}.

Our analysis considers a reindeer pellet-group survey that was conducted on Storliden Mountain in the northern forest area of Sweden over the years 2009–2010. We focus in 2009 and specifically we consider 
pellet-group count data collected between June 3rd and 8th in 2009 \citep{Lee2016}. The main goal of the survey was to assess the impact of newly established wind farms on reindeer habitats. 


In practice, we observe the total number of pellet-groups (a pellet-group is defined as a cluster of 20 or more pellets)
at 357 location sites $y(\bs_i)$, $i=1,2,\ldots, 357$.
In this case, the mechanism generating the pellet-groups, for each location, can be assumed as a Poisson process, where the  inter-arrival times between one pellet-group and another could be assumed to be exponentially distributed. Under this assumption,
the proposed Poisson random field can be a useful tool to analyze the pellet-groups counts data.

The dataset possesses two challenging features. The first one is that 73.67\% of the counts are zeros since the animal might move as it defecates, and some plots present zero pellet-group counts (see Figure \ref{fig:map}). The second one is that the empirical semi-variogram (see Figure 3 of the online supplement)
exhibits both spatial correlation and a nugget effect.

\begin{figure}[ht!]
\begin{center}
\setlength{\unitlength}{0.1\textwidth}
\scalebox{0.65}{
\includegraphics[width=0.8\linewidth, height=0.4\textheight]{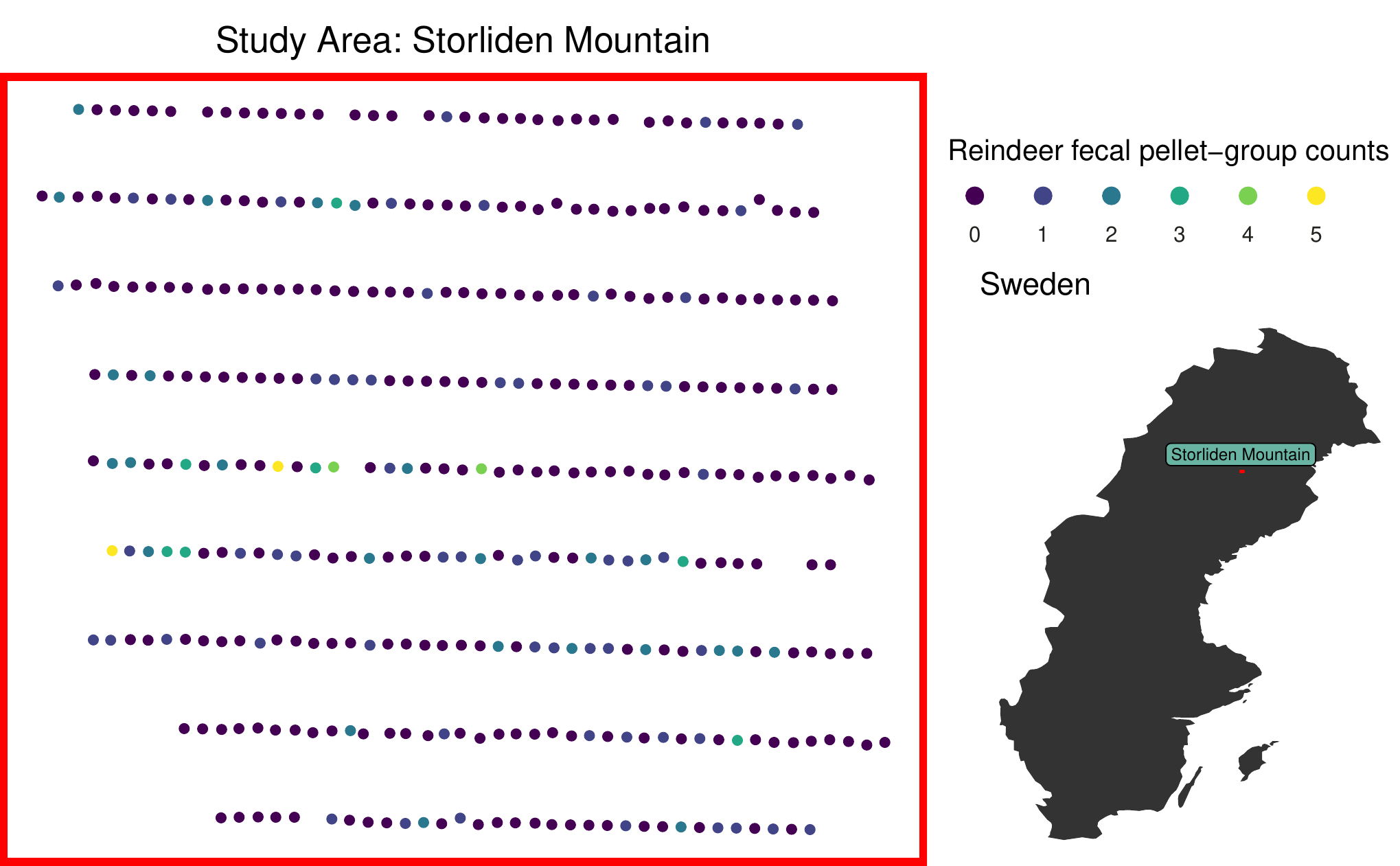} }
\end{center}
\caption{Spatial location of reindeer pellet-group survey data.
}\label{fig:map}
\end{figure}

To face this problem we consider the 
zero-inflated Poisson (ZIP) random field proposed in Section \ref{subsec:3.3}.
and we compare it with the ZIP Gaussian copula (ZIP GC) using the 
\texttt{R} package \texttt{gcKrig} \citep{JSSv087i13}. In addition, we consider 
the ZIP Log-Gaussian (ZIP LG) random field  as implemented in the \texttt{R} package \texttt{INLA} \citep{R-INLA:1, R-INLA:2, R-INLA:3} which exploits the integrated nested Laplace approximation, 
under a Bayesian framework, in the estimation step. 

Since our application's primary goal is to assess the impact of newly established wind farms on reindeer habitats, we are interested in relating the number of pellet-groups with covariates such as distance to power lines, slopes, or elevation of the field. Therefore, and following the results of \cite{Lee2016}, we include three covariates: Northwest slopes (NS), Elevation (Eln) and Distance to power lines (DPL). In particular we specify $\lambda(\bs)$ as:
\vspace{-0.5cm}
$$\lambda(\bs)=\exp(\beta_0+\beta_{NS}\mathrm{NS}(\bs)+\beta_{Eln}\mathrm{Eln}(\bs)+\beta_{DPL}\mathrm{DPL}(\bs)).\vspace{-0.5cm}$$
The parameterization for the marginal mean and variance is slightly different for the three models.
Specifically, assuming that the  probability of excess of zero counts $p$ does not depend on $\bs$, the marginal mean and variance specifications 
are given by $\E(Y(\bs))=\lambda(\bs)(1-p)$ and  $\Var(Y(\bs))=\E(Y(\bs))\left[1+\dfrac{p}{1-p}\E(Y(\bs))\right]$ for the proposed model, 
$\E(N(\bs))=\lambda(\bs)$ and $\Var(Y(\bs))=\E(Y(\bs))\left[1+\theta_{GC}\E(Y(\bs))\right]$ for the GC model,  $\E(Y(\bs))=\lambda(\bs)\exp(0.5\sigma^2)(1-p)$ and  $\Var(Y(\bs))=\E(Y(\bs))\left[1+\dfrac{p}{1-p}\E(Y(\bs))\right]+\dfrac{\E(Y(\bs))^2}{1-p}\left[\exp(\sigma^2)-1\right]$ for the LG model.

Here, 
$p$ is specified as $\Phi(\theta)$, with $\theta \in \R$, as $\frac{\theta_{GC}}{1+\theta_{GC}}$, 
with $\theta_{GC}>0$, and  as $\frac{\exp(\theta_{LG})}{1+\exp(\theta_{LG})}$ , with $\theta_{LG} \in \R$, respectively,
so $\theta$, $\theta_{GC}$, $\theta_{LG}$ can be interpreted as overdispersion parameters.
 It is important to remark that $\beta_0$ can not be compared between the different approaches, but $\beta_{NS}$, $\beta_{Eln}$ and $\beta_{DPL}$ can be compared.     

We assume an underlying exponential correlation model with nugget effect $\rho(\bh)=(1-\tau^2)e^{-||\bh||/\alpha} +\tau^2\mathds{1}_{0}(||\bh||)$ for the ZIP GC and ZIP LG random fields.
On the other hand for the proposed ZIP model 
we specify $\rho_1(\bh)=(1-\tau_2^2)e^{-||\bh||/\alpha}+\tau_2^2\mathds{1}_{0}(||\bh||)$ and $\rho_2(\bh)=(1-\tau_1^2)e^{-||\bh||/\alpha}+\tau_1^2\mathds{1}_{0}(||\bh||)$
that is two different underlying correlation models for $B$ and $N$ in \eqref{ppppp}
sharing a common exponential correlation model and  different nugget effects $\tau_1^2,\tau_2^2$.

We use maximum {\it wpl} estimation with $\xi=150$ in \eqref{wer} for our ZIP random field. 
For the ZIP GC model, we perform maximum likelihood estimation as explained in Section 3 in the online supplement, and for ZIP LG model, we perform approximate Bayesian inference 
 using the \texttt{INLA} approach \citep{R-INLA:1}.

Table \ref{est.app} summarizes the results of the estimates, including their standard error
for the three models. In the case of our ZIP random field, standard errors were computed by using parametric bootstrap \citep{Efron.Tibshirani::1986}. For the Poisson LG model the reported estimates are the  means of the posterior distributions with associated standard error.

Note that if $\beta_{DPL}$ is a positive value, then the counts of pellet-groups increase at larger distances from the power lines, {\it i.e.}, there is a greater reindeer population far from the wind farms. The estimation of the regression parameters is quite similar for our ZIP and the ZIP GC models with lower standard error estimations for the GC model. On the other hand, our ZIP model shows the smallest standard error estimation of the spatial scale parameter $\alpha$. Finally, the estimates of $p$, the excess of zero counts, which depend on $\theta$, $\theta_{GC}$
and $\theta_{LG}$ for the ZIP, ZIP GC and ZIP LG random fields, are given by $0.481$, $0.477$, $0.573$, respectively.

\begin{table}[ht!]
\centering
\scalebox{0.7}{
\begin{tabular}{|c|c|c|c|c|c|c|c|c|c|c|c|c|c|}
  \cline{2-14}
\multicolumn{1}{c|}{} & $\beta_0$  &  $\beta_{NS}$ & $\beta_{Eln}$ &  $\beta_{DPL}$ &$\theta$ & $\theta_{GC}$& $\theta_{LG}$ &$\alpha$ & $\tau_1^2$ & $\tau_2^2$ & $\tau^2$& $\sigma^2$ & $\overline{\mathrm{RMSE}}$ \\ 
  \hline
\multirow{ 2}{*}{ZIP}    & -23.423      & -0.534   & 0.005        & 2.594        & -0.048   &               &
                &339.132 &0.868   & 0.624    &           &               & \multirow{ 2}{*}{0.797}\\ 
                    & (6.473)      & (0.469) & (0.004)       & (0.742) & (1.048)       &               &
                &(92.981)          & (0.263) & (0.308)       &           &               &  \\
\hline
\multirow{ 2}{*}{ZIP GC}    &-19.096       & -0.465   &0.003         &2.060     &               &0.912         &
                    &298.926  &     &          &0.714         &           &\multirow{ 2}{*}{0.800}\\
                    &(2.309)       & (0.364)   &(0.003)         &(0.265)     &               &(0.261)         &
                    &(186.990)  &     &          &(0.1285)         &           &\\
\hline
\multirow{ 2}{*}{ZIP LG}     &-17.905       &-0.826    &0.010         &1.622     &               &               &
 0.293             &685.922   &       &           & 0.084    & 0.735    &\multirow{ 2}{*}{0.835} \\
                    &(9.514)       &(0.388)    &(0.005)         &(1.177)     &               &               &
 (0.098)             &(354.614)   &       &           & (0.022)    & (0.396)    & \\
   \hline
\end{tabular}}
\caption{Parameter estimates for the reindeer pellet-group survey data obtained under the ZIP, ZIP GC and ZIP LG random fields. The associated standard errors are in parenthesis. The last column shows the associated empirical mean of the RMSE for each model.} \label{est.app}
\end{table}

The three models considered can also be used for prediction of the of pellet-group counts at specific not sampled location sites. In particular, this approach allows us to discover new areas for reindeer habitat employing the predicted number of pellet-groups, potentially providing information about the behavior of the reindeer population over the entire region.
With this goal in mind, we want to assess the predictive performances of the three models. To do so, we randomly choose 80\% of the spatial locations ({\it i.e.}, $286$ location sites) for the parameter estimation and use the remaining 20\% ({\it i.e.}, $71$ location sites) for the predictions. We repeat this procedure $100$ times, recording the RMSE each time. Specifically, for each $j$-th left-out sample $(y_j(\bs_1),y_j(\bs_2),\ldots,y_j(\bs_{71}))$, we compute
$$\mathrm{RMSE_j}=\left(\dfrac{1}{71}\sum\limits_{i=1}^{71}(y_j(\bs_i)-\widehat{Y}_j(\bs_i))^2\right)^{1/2},$$
where
 $\widehat{Y}_j(\bs_i)$ is the optimal linear predictor for our ZIP random field (computed using the correlation given in Proposition 1 in the online supplement), the optimal predictor for the ZIP GC random field and the mean of the posterior predictive distribution for the ZIP LG random field. We report in Table \ref{est.app}, the empirical mean of the RMSE obtained for each left-out sample, {\it i.e.}, $\overline{\mathrm{RMSE}}=\sum\limits_{j=1}^{100}\mathrm{RMSE_j}/100$. The ZIP and ZIP GC random fields' clearly outperform the ZIP LG random field in terms of prediction performance. In particular the  proposed ZIP random field provides the smallest $\overline{\mathrm{RMSE}}$.

 Finally, as suggested by one Referee, we perform a simulation-based model assessment. Specifically, we simulated 10000 realizations under the three fitted models and counted the number of observations lying between the 95\% probability intervals constructed with the simulated data. The results show that our ZIP proposed model and the ZIP GC are pretty similar, with 97.2\% and 97.7\% of the data lying in the 95\% probability intervals, respectively. In the case of the ZIP LG, 96.9\% of observations lie between the 95\% probability interval.

In addition, we compared the empirical semi-variogram of the data with the ones
obtained from the simulations, and we found that our proposal (ZIP model) presents a small 95\% probability interval compared to its competitors (a tight interval) (see Figure \ref{fig:semivars}). This situation implies that our approach provides less uncertainty when estimating spatial dependence.

\begin{figure}[ht!]
\begin{center}
\setlength{\unitlength}{0.1\textwidth}
\scalebox{0.5}{
\includegraphics[scale=1]{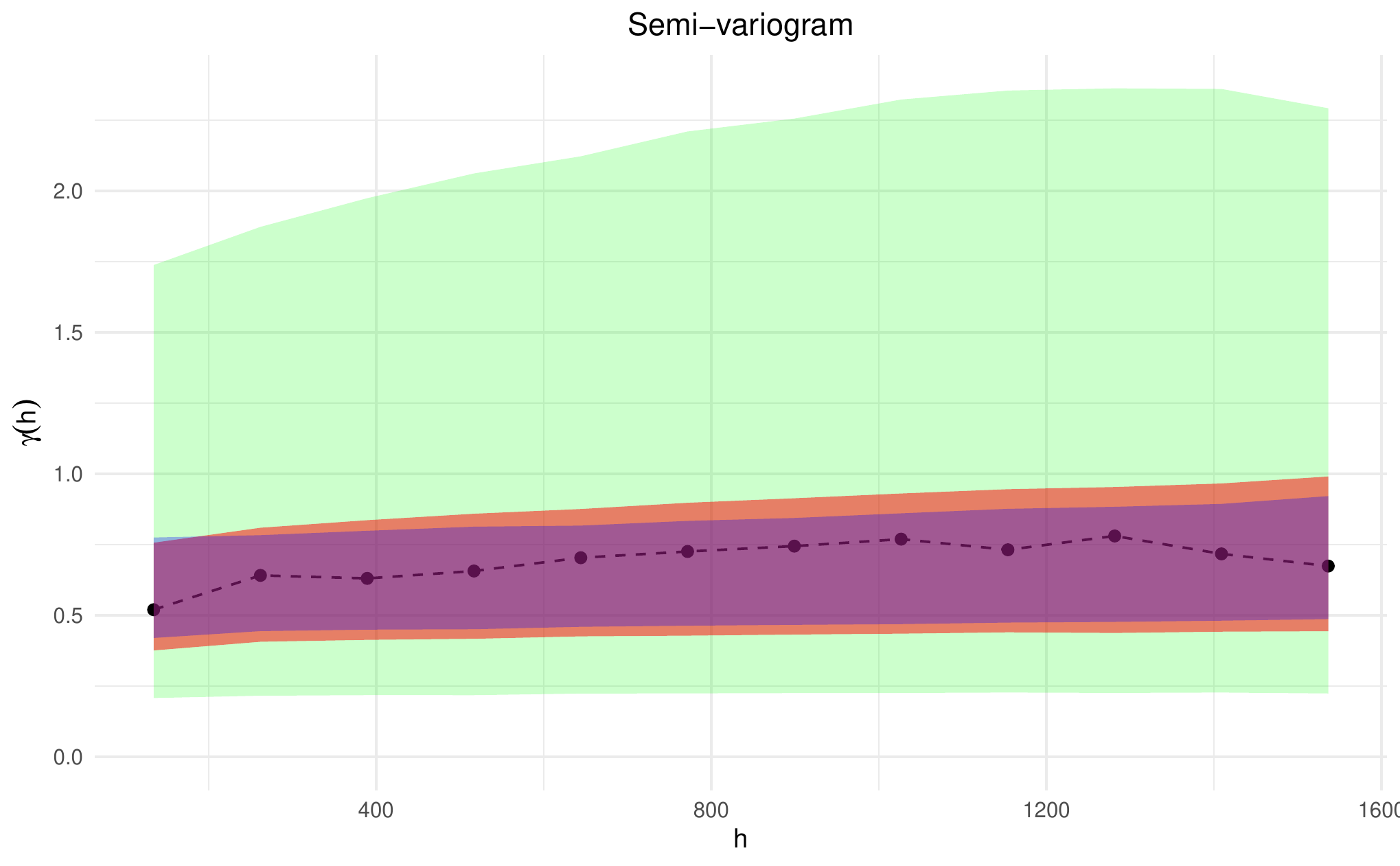} }
\end{center}
\caption{95\% probability intervals for semi-variograms under ZIP (purple), ZIP GC (orange) and ZIP LG (green) models. The empirical semi-variogram for the pellet group data is depicted in the black dashed line.
}\label{fig:semivars}
\end{figure}

\section{Concluding remarks}\label{sec:7}

This paper has introduced a  novel Poisson random field, {\it i.e.}, a random field with Poisson marginal distributions, for regression and dependence analysis when addressing point-referenced count data defined on a spatial Euclidean space. However, the proposed methodology can be easily adapted to other types of data, such as space-time \citep{gneiting2013}, areal \citep{Rue:Held:2005} or spherical data \citep{gneiting2013}.

By construction, for each spatial location, the proposed model is a Poisson counting process, {\it i.e.}, it represents 
the random total number of events occurring in an arbitrary interval of time when the inter-arrival times are exponentially distributed.

More importantly, given two arbitrary location sites, the associated Poisson counting processes are spatially correlated. For this reason, the proposed model can be viewed as a spatial generalization of the Poisson process.

The correlation between the Poisson counting processes is achieved by considering  sequences of independent copies of a random field with an exponential marginal distribution as inter-arrival times in the counting renewal processes framework. The resulting (non-)stationary random field is marginally Poisson distributed and the dependence is indexed by a correlation function.

They key features of the proposed Poisson random field  with respect to the  Poisson Log-Gaussian random field are that its marginal distribution is Poisson distributed and it can be mean square continuous or not. The Poisson Gaussian copula approach shares these good features with our model. However, the generating mechanisms ({\it i.e.}, the Poisson process), underlying our model makes it more appealing from interpretability viewpoint. 

In our proposal, a possible limitation is that inference based on full likelihood cannot be performed due to the lack of amenable expressions of the associated multivariate distributions. Nevertheless, the simulations studies we conducted showed that our approach, 
based on a pairwise likelihood estimation seems to be an effective solution for estimating the unknown parameters involved in the Poisson random field. Another potential limitation is that the optimal predictor that minimizes the mean square prediction error is not available. However, our numerical experiments show that our solution based on optimal linear predictor performs very well when compared with the optimal predictors of the Poisson Gaussian copula and Poisson Log Gaussian models.
Finally, the application of our model to the reindeer pellet-group survey data in Sweden shows that our approach can be easily adapted to handle spatial count data with an excessive number of zeros.

A well-known restriction of the Poisson distribution is equidispersion. Unfortunately, this situation is not always observed in real spatial data. The class of random fields proposed in \eqref{qqq} can be used to obtain random fields with flexible marginal models that consider over or under dispersion. In this case, a possible solution is to consider random fields with a more flexible marginal distribution than the exponential 
marginal distribution, such as the gamma or Weibull random fields \citep{Bevilacqua:2018ab}. The resulting marginal counting models have been studied in \cite{wink95} and \cite{mac2008}.
 Another alternative for obtaining over dispersed random fields is considering scale mixtures of Poisson random fields. These topics are currently under study and will be included in a forthcoming paper.

\if1\blind
{
\section*{Acknowledgements}

The work of Diego Morales-Navarrete, Moreno Bevilacqua and Luis M. Castro was partially supported by ANID - Millennium Science Initiative Program - NCN17\_059 from the Chilean government. Diego Morales-Navarrete also acknowledges support from the VRI-UC scholarship from the Pontificia Universidad Cat\'olica de Chile. Moreno Bevilacqua acknowledges financial support from Grant FONDECYT 1200068 and ANID/PIA/ANILLOS ACT210096, Chile 
and  project Data Observatory Foundation DO210001
from the Chilean government. The work of Christian Caama\~no was partially supported by grant FONDECYT 11220066 from the Chilean government and DIUBB 2120538 IF/R from University of B\'io-B\'io. Luis M. Castro acknowledges financial support from Grant FONDECYT 1220799 from the Chilean government. 
} \fi

\if0\blind
{

} \fi

\bibliographystyle{biom}

\bibliography{mybib2}

\begin{thebibliography}{}

\bibitem[\protect\citeauthoryear{Aitchison and Ho}{Aitchison and
  Ho}{1989}]{Aitchison}
Aitchison, J. and Ho, C.~H. (1989).
\newblock The multivariate {P}oisson-{L}og normal distribution.
\newblock {\em Biometrika} {\bf 76,} 643--653.

\bibitem[\protect\citeauthoryear{Alegr\'ia, Caro, Bevilacqua, Porcu, and
  Clarke}{Alegr\'ia et~al.}{2017}]{ALEGRIA2017}
Alegr\'ia, A., Caro, S., Bevilacqua, M., Porcu, E., and Clarke, J. (2017).
\newblock Estimating covariance functions of multivariate skew-gaussian random
  fields on the sphere.
\newblock {\em Spatial Statistics} {\bf 22,} 388 -- 402.

\bibitem[\protect\citeauthoryear{Bai, Kang, and Song}{Bai
  et~al.}{2014}]{Bai:Kang:Song:2014}
Bai, Y., Kang, J., and Song, P. (2014).
\newblock Efficient pairwise composite likelihood estimation for
  spatial-clustered data.
\newblock {\em Biometrics} {\bf 7,} 661--670.

\bibitem[\protect\citeauthoryear{Banerjee, Carlin, and Gelfand}{Banerjee
  et~al.}{2014}]{Banerjee-Carlin-Gelfand:2014}
Banerjee, S., Carlin, B.~P., and Gelfand, A.~E. (2014).
\newblock {\em Hierarchical Modeling and Analysis for Spatial Data}.
\newblock Chapman \& Hall/CRC Press, Boca Raton: FL.

\bibitem[\protect\citeauthoryear{Bennett, English, and McCain}{Bennett
  et~al.}{1940}]{BennetEtAl1940}
Bennett, L., English, P., and McCain, R. (1940).
\newblock A study of deer populations by use of pellet-group counts.
\newblock {\em The Journal of Wildlife Management} {\bf 4,} 398–403.

\bibitem[\protect\citeauthoryear{Bevilacqua, Caama{\~n}o-Carrillo, and
  Gaetan}{Bevilacqua et~al.}{2020}]{Bevilacqua:2018ab}
Bevilacqua, M., Caama{\~n}o-Carrillo, C., and Gaetan, C. (2020).
\newblock On modelling positive continuous data with spatio-temporal
  dependence.
\newblock {\em Environmetrics} {\bf 31,} e2632.

\bibitem[\protect\citeauthoryear{Bevilacqua, Caamaño-Carrillo, Arellano-Valle,
  and Morales-Oñate}{Bevilacqua et~al.}{2021}]{Bev:2020}
Bevilacqua, M., Caamaño-Carrillo, C., Arellano-Valle, R., and Morales-Oñate,
  V. (2021).
\newblock Non-gaussian geostatistical modeling using (skew) t processes.
\newblock {\em Scandinavian Journal of Statistics} {\bf 48,} 212--245.

\bibitem[\protect\citeauthoryear{Bevilacqua, Faouzi, Furrer, and
  Porcu}{Bevilacqua et~al.}{2019}]{Bevilacqua:Faouzi_et_all:2019}
Bevilacqua, M., Faouzi, T., Furrer, R., and Porcu, E. (2019).
\newblock Estimation and prediction using generalized wendland functions under
  fixed domain asymptotics.
\newblock {\em The Annals of Statistics} {\bf 47,} 828--856.

\bibitem[\protect\citeauthoryear{Bevilacqua and Gaetan}{Bevilacqua and
  Gaetan}{2015}]{Bevilacqua:Gaetan:2015}
Bevilacqua, M. and Gaetan, C. (2015).
\newblock Comparing composite likelihood methods based on pairs for spatial
  {G}aussian random fields.
\newblock {\em Statistics and Computing} {\bf 25,} 877--892.

\bibitem[\protect\citeauthoryear{Bevilacqua, Gaetan, Mateu, and
  Porcu}{Bevilacqua et~al.}{2012}]{Bevilacqua:Gaetan:Mateu:Porcu:2012}
Bevilacqua, M., Gaetan, C., Mateu, J., and Porcu, E. (2012).
\newblock Estimating space and {space-time} covariance functions for large data
  sets: a weighted composite likelihood approach.
\newblock {\em Journal of the American Statistical Association} {\bf 107,}
  268--280.

\bibitem[\protect\citeauthoryear{Bevilacqua, Morales-Oñate, and
  Caamaño-Carrillo}{Bevilacqua et~al.}{2022}]{Bevilacqua:2018aa}
Bevilacqua, M., Morales-Oñate, V., and Caamaño-Carrillo, C. (2022).
\newblock {\em GeoModels: Procedures for Gaussian and Non Gaussian
  Geostatistical (Large) Data Analysis}.
\newblock R package version 1.0.2.

\bibitem[\protect\citeauthoryear{Buckland, Anderson, Burnham, Laake, Borchers,
  and Thomas}{Buckland et~al.}{2001}]{BucklandEtAl2001}
Buckland, S., Anderson, D.~R., Burnham, K.~P., Laake, J.~L., Borchers, D.~L.,
  and Thomas, L. (2001).
\newblock {\em Introduction to Distance Sampling-estimating Abundance of
  Biological Populations.}
\newblock New York, NY., Oxford University Press.

\bibitem[\protect\citeauthoryear{Christensen and Waagepetersen}{Christensen and
  Waagepetersen}{2002}]{Christensen:Waagepetersen:2002}
Christensen, O. and Waagepetersen (2002).
\newblock Bayesian prediction of spatial count data using generalized linear
  mixed models.
\newblock {\em Biometrics} {\bf 58,} 280--286.

\bibitem[\protect\citeauthoryear{Cox}{Cox}{1962}]{Cox:1970}
Cox, D. (1962).
\newblock {\em Renewal Theory}.
\newblock Methuen \& Co, London.

\bibitem[\protect\citeauthoryear{Cressie and Wikle}{Cressie and
  Wikle}{2011}]{Cressie:Wikle:2011}
Cressie, N. and Wikle, C. (2011).
\newblock {\em Statistics for Spatio-Temporal Data.}
\newblock Wiley Series in Probability and Statistics. Wiley.

\bibitem[\protect\citeauthoryear{{De Oliveira}}{{De
  Oliveira}}{2006}]{DeOliveira:2006}
{De Oliveira}, V. (2006).
\newblock On optimal point and block prediction in log-{G}aussian random
  fields.
\newblock {\em Scandinavian Journal of Statistics} {\bf 33,} 523--540.

\bibitem[\protect\citeauthoryear{{De Oliveira}}{{De
  Oliveira}}{2013}]{Oliveira:2013}
{De Oliveira}, V. (2013).
\newblock Hierarchical {P}oisson models for spatial count data.
\newblock {\em Journal of Multivariate Analysis} {\bf 122,} 93--408.

\bibitem[\protect\citeauthoryear{{De Oliveira}}{{De
  Oliveira}}{2014}]{Oliveira:2014}
{De Oliveira}, V. (2014).
\newblock Poisson kriging: A closer investigation.
\newblock {\em Spatial Statistics} {\bf 7,} 1--20.

\bibitem[\protect\citeauthoryear{Diggle and Giorgi}{Diggle and
  Giorgi}{2019}]{Diggle-Giorgi:2019}
Diggle, P. and Giorgi, E. (2019).
\newblock {\em Model-based Geostatistics for Global Public Health: Methods and
  Applications}.
\newblock Chapman \&amp; Hall/CRC Interdisciplinary Statistics. Chapman and
  Hall/CRC Press.

\bibitem[\protect\citeauthoryear{Diggle, Tawn, and Moyeed}{Diggle
  et~al.}{1998}]{Diggle:Tawn:Moyeed:1998}
Diggle, P., Tawn, J., and Moyeed, R. (1998).
\newblock Model-based geostatistics.
\newblock {\em Journal of the Royal Statistical Society: Series C {(Applied}
  Statistics)} {\bf 47,} 299--350.

\bibitem[\protect\citeauthoryear{Diggle and Ribeiro}{Diggle and
  Ribeiro}{2007}]{Diggle-Ribeiro:2007}
Diggle, P.~J. and Ribeiro, P.~J. (2007).
\newblock {\em Model-based Geostatistics}.
\newblock Springer, New York.

\bibitem[\protect\citeauthoryear{Eberhardt and Van~Etten}{Eberhardt and
  Van~Etten}{1956}]{EberhardtVanEtten1956}
Eberhardt, L. and Van~Etten, R. (1956).
\newblock {Evaluation of the pellet group count as a deer census method}.
\newblock {\em Journal of Wildlife Management} {\bf 20,} 70--74.

\bibitem[\protect\citeauthoryear{Efron and Tibshirani}{Efron and
  Tibshirani}{1986}]{Efron.Tibshirani::1986}
Efron, B. and Tibshirani, R. (1986).
\newblock Bootstrap methods for standard errors, confidence intervals, and
  other measures of statistical accuracy.
\newblock {\em Statistical Science} {\bf 1,} 54--75.

\bibitem[\protect\citeauthoryear{Etten and Bennett}{Etten and
  Bennett}{1965}]{VanEttenBennet1965}
Etten, R.~V. and Bennett, C. (1965).
\newblock Some sources of error in using pellet-group counts for censusing
  deer.
\newblock {\em The Journal of Wildlife Management} {\bf 29,} 723–729.

\bibitem[\protect\citeauthoryear{Freddy and Bowden}{Freddy and
  Bowden}{1983}]{FreedyBowden1983}
Freddy, D. and Bowden, D. (1983).
\newblock {Sampling mule deer pellet-group densities in juniper-pinyon
  woodland}.
\newblock {\em Journal of Wildlife Management} {\bf 47,} 476--485.

\bibitem[\protect\citeauthoryear{Gelfand and Schliep}{Gelfand and
  Schliep}{2016}]{GELFAND201686}
Gelfand, A.~E. and Schliep, E.~M. (2016).
\newblock Spatial statistics and {G}aussian processes: A beautiful marriage.
\newblock {\em Spatial Statistics} {\bf 18,} 86--104.
\newblock Spatial Statistics Avignon: Emerging Patterns.

\bibitem[\protect\citeauthoryear{Genest and Nešlehová}{Genest and
  Nešlehová}{2007}]{Genest:Neslehova:2007}
Genest, C. and Nešlehová, J. (2007).
\newblock A primer of copulas for count data.
\newblock {\em ASTIN Bulletin} {\bf 37,} 475--515.

\bibitem[\protect\citeauthoryear{Gneiting}{Gneiting}{2002}]{Gneiting:2002}
Gneiting, T. (2002).
\newblock Stationary covariance functions for space-time data.
\newblock {\em Journal of the American Statistical Association} {\bf 97,}
  590--600.

\bibitem[\protect\citeauthoryear{Gneiting}{Gneiting}{2013}]{gneiting2013}
Gneiting, T. (2013).
\newblock Strictly and non-strictly positive definite functions on spheres.
\newblock {\em Bernoulli} {\bf 19,} 1327--1349.

\bibitem[\protect\citeauthoryear{Gough}{Gough}{2009}]{gough2009gnu}
Gough, B. (2009).
\newblock {\em {GNU Scientific Library Reference Manual}}.
\newblock Network Theory Ltd.

\bibitem[\protect\citeauthoryear{Gouriéroux, Monfort, and Renault}{Gouriéroux
  et~al.}{2017}]{cppp}
Gouriéroux, C., Monfort, A., and Renault, E. (2017).
\newblock Consistent pseudo-maximum likelihood estimators.
\newblock {\em Annals of Economics and Statistics} pages 187--218.

\bibitem[\protect\citeauthoryear{Gradshteyn and Ryzhik}{Gradshteyn and
  Ryzhik}{2007}]{Gradshteyn:Ryzhik:2007}
Gradshteyn, I. and Ryzhik, I. (2007).
\newblock {\em Table of Integrals, Series, and Products}.
\newblock Academic Press, New York, 7 edition.

\bibitem[\protect\citeauthoryear{Guillot, Loren, and Rudemo}{Guillot
  et~al.}{2009}]{Guillot_et_all:2009}
Guillot, G., Loren, N., and Rudemo, M. (2009).
\newblock Spatial prediction of weed intensities from exact count data and
  image-based estimates.
\newblock {\em Journal of the Royal Statistical Society: Series C} {\bf 58,}
  525--542.

\bibitem[\protect\citeauthoryear{Han and {De Oliveira}}{Han and {De
  Oliveira}}{2016}]{Han:Oliveira:2016}
Han, Z. and {De Oliveira}, V. (2016).
\newblock On the correlation structure of {G}aussian copula models for
  geostatistical count data.
\newblock {\em Australian \& New Zealand Journal of Statistics} {\bf 58,}
  47--69.

\bibitem[\protect\citeauthoryear{Han and Oliveira}{Han and
  Oliveira}{2018}]{JSSv087i13}
Han, Z. and Oliveira, V.~D. (2018).
\newblock gckrig: An {R} package for the analysis of geostatistical count data
  using {G}aussian copulas.
\newblock {\em Journal of Statistical Software} {\bf 87,} 1--32.

\bibitem[\protect\citeauthoryear{Heagerty and Lele}{Heagerty and
  Lele}{1998}]{Heagerty:Lele:1998}
Heagerty, P. and Lele, S. (1998).
\newblock A composite likelihood approach to binary spatial data.
\newblock {\em Journal of the American Statistical Association} {\bf 93,} 1099
  --1111.

\bibitem[\protect\citeauthoryear{Hunter}{Hunter}{1974}]{Hunter:1974}
Hunter, J. (1974).
\newblock Renewal theory in two dimensions: basic results.
\newblock {\em Advances in Applied Probability} {\bf 6,} 376--391.

\bibitem[\protect\citeauthoryear{Joe}{Joe}{2014}]{Joe:2014}
Joe, H. (2014).
\newblock {\em Dependence {M}odeling with {C}opulas}.
\newblock Chapman and Hall/CRC, Boca Raton, FL.

\bibitem[\protect\citeauthoryear{Joe and Lee}{Joe and Lee}{2009}]{Joe:Lee:2009}
Joe, H. and Lee, Y. (2009).
\newblock On weighting of bivariate margins in pairwise likelihood.
\newblock {\em Journal of Multivariate Analysis} {\bf 100,} 670--685.

\bibitem[\protect\citeauthoryear{Kazianka and Pilz}{Kazianka and
  Pilz}{2010}]{Kazianka:Pilz:2010}
Kazianka, H. and Pilz, J. (2010).
\newblock Copula-based geostatistical modeling of continuous and discrete data
  including covariates.
\newblock {\em Stochastic Environmental Research and Risk Assessment} {\bf 24,}
  661--673.

\bibitem[\protect\citeauthoryear{Kibble}{Kibble}{1941}]{Kibble:1941}
Kibble, W.~F. (1941).
\newblock A two-variate {gamma} type distribution.
\newblock {\em Sankhy\={a}: The Indian Journal of Statistics} {\bf 5,}
  137--150.

\bibitem[\protect\citeauthoryear{Krebs, Boonstra, Nams, M, Hodges, and
  Boutin}{Krebs et~al.}{2001}]{KrebsEtAl2001}
Krebs, C., Boonstra, R., Nams, V., M, M.~O., Hodges, K., and Boutin, S. (2001).
\newblock Estimating snowshoe hare population density from pellet plots: A
  further evaluation.
\newblock {\em Canadian Journal of Zoology} {\bf 79,} 1–4.

\bibitem[\protect\citeauthoryear{Krishnaiah and Rao}{Krishnaiah and
  Rao}{1961}]{Krishnaiah:Rao:1961}
Krishnaiah, P. and Rao, M. (1961).
\newblock Remarks on a multivariate gamma distribution.
\newblock {\em The American Mathematical Monthly} {\bf 68(4),} 342--346.

\bibitem[\protect\citeauthoryear{Krishnamoorthy and
  Parthasarathy}{Krishnamoorthy and
  Parthasarathy}{1951}]{Krishnamoorthy:Parthasarathy:1951}
Krishnamoorthy, A.~S. and Parthasarathy, M. (1951).
\newblock A {M}ultivariate {G}amma-{T}ype {D}istribution.
\newblock {\em The Annals of Mathematical Statistics} {\bf 22(4),} 549--557.

\bibitem[\protect\citeauthoryear{Lee, Moudud, Noh, Rönnegård, and Skarin}{Lee
  et~al.}{2016}]{Lee2016}
Lee, Y., Moudud, A., Noh, M., Rönnegård, L., and Skarin, A. (2016).
\newblock Spatial modeling of data with excessive zeros applied to reindeer
  pellet‐group counts.
\newblock {\em Ecology and Evolution} {\bf 6,} 7047--7056.

\bibitem[\protect\citeauthoryear{Lindgren, Rue, , and Lindstrom}{Lindgren
  et~al.}{2011}]{R-INLA:2}
Lindgren, F., Rue, H., , and Lindstrom, J. (2011).
\newblock An explicit link between {G}aussian fields and {G}aussian {M}arkov
  random fields: the stochastic partial differential equation approach.
\newblock {\em Journal of the Royal Statistical Society B} {\bf 73,} 423--498.

\bibitem[\protect\citeauthoryear{Lindsay}{Lindsay}{1988}]{Lindsay:1988}
Lindsay, B. (1988).
\newblock Composite likelihood methods.
\newblock {\em Contemporary Mathematics} {\bf 80,} 221--239.

\bibitem[\protect\citeauthoryear{Mainardi, Gorenflo, and Vivoli}{Mainardi
  et~al.}{2007}]{MAINARDI2007725}
Mainardi, F., Gorenflo, R., and Vivoli, A. (2007).
\newblock Beyond the {P}oisson renewal process: A tutorial survey.
\newblock {\em Journal of Computational and Applied Mathematics} {\bf 205,} 725
  -- 735.
\newblock Special issue on evolutionary problems.

\bibitem[\protect\citeauthoryear{Marques, Buckland, Goffin, Dixon, Borchers,
  Mayle, and Peace}{Marques et~al.}{2001}]{MarquesEtAl2001}
Marques, F., Buckland, S., Goffin, D., Dixon, C., Borchers, D., Mayle, B., and
  Peace, A. (2001).
\newblock {Estimating deer abundance from line transect surveys of dung: sika
  deer in sourthern Scotland}.
\newblock {\em Journal of Applied Ecology} {\bf 38,} 349--363.

\bibitem[\protect\citeauthoryear{Martins, Simpson, Lindgren, and Rue}{Martins
  et~al.}{2013}]{R-INLA:3}
Martins, T.~G., Simpson, D., Lindgren, F., and Rue, H. (2013).
\newblock Bayesian computing with inla: new features.
\newblock {\em Computational Statistics and Data Analysis} {\bf 67,} 68--83.

\bibitem[\protect\citeauthoryear{Masarotto and Varin}{Masarotto and
  Varin}{2012}]{Masarotto:Varin:2012}
Masarotto, G. and Varin, C. (2012).
\newblock Gaussian copula marginal regression.
\newblock {\em Electronic Journal of Statistics} {\bf 6,} 1517--1549.

\bibitem[\protect\citeauthoryear{Masuda}{Masuda}{2013}]{AOS1121}
Masuda, H. (2013).
\newblock {Convergence of Gaussian quasi-likelihood random fields for ergodic
  Lévy driven SDE observed at high frequency}.
\newblock {\em The Annals of Statistics} {\bf 41,} 1593--1641.

\bibitem[\protect\citeauthoryear{Mayle, Peace, and Gill}{Mayle
  et~al.}{1999}]{MayleEtAl1999}
Mayle, B., Peace, A., and Gill, R. (1999).
\newblock How many deer? a field guide to estimating deer population size.
\newblock Technical Report Forestry Commission Field book 18, Forestry
  Commission, Edinburgh.

\bibitem[\protect\citeauthoryear{McShane, Adrian, Bradlow, and Fader}{McShane
  et~al.}{2008}]{mac2008}
McShane, B., Adrian, M., Bradlow, E.~T., and Fader, P.~S. (2008).
\newblock Count models based on weibull interarrival times.
\newblock {\em Journal of Business \& Economic Statistics} {\bf 26,} 369--378.

\bibitem[\protect\citeauthoryear{Monestiez, Dubroca, Bonnin, and
  Guinet}{Monestiez et~al.}{2006}]{Monestiez_et_al2006}
Monestiez, P., Dubroca, L., Bonnin, E., and Guinet, J.-P. D.~C. (2006).
\newblock Geostatistical modelling of spatial distribution of balaenoptera
  physalus in the northwestern mediterranean sea from sparse count data and
  heterogeneous.
\newblock {\em Ecological Modelling} {\bf 193,} 615--628.

\bibitem[\protect\citeauthoryear{Mooty, Karns, and Heisey}{Mooty
  et~al.}{1984}]{MootyEtAl1984}
Mooty, J., Karns, P., and Heisey, D. (1984).
\newblock {The relationship between white-tailed deer track counts and
  pellet-group surveys}.
\newblock {\em Journal of Wildlife Management} {\bf 48,} 275--279.

\bibitem[\protect\citeauthoryear{Porcu, Bevilacqua, and Genton}{Porcu
  et~al.}{2016}]{porcubev}
Porcu, E., Bevilacqua, M., and Genton, M.~G. (2016).
\newblock Spatio-temporal covariance and cross-covariance functions of the
  great circle distance on a sphere.
\newblock {\em Journal of the American Statistical Association} {\bf 111,}
  888--898.

\bibitem[\protect\citeauthoryear{{R Core Team}}{{R Core Team}}{2020}]{R2020}
{R Core Team} (2020).
\newblock {\em R: A Language and Environment for Statistical Computing}.
\newblock R Foundation for Statistical Computing, Vienna, Austria.

\bibitem[\protect\citeauthoryear{Ross}{Ross}{2008}]{ross2008stochastic}
Ross, S. (2008).
\newblock {\em Stochastic {P}rocesses, 2nd. Ed.}
\newblock Wiley series in probability and mathematical statistics. Wiley India
  Pvt. Limited.

\bibitem[\protect\citeauthoryear{Rowland, White, and Karlen}{Rowland
  et~al.}{1984}]{RowlandEtAl1984}
Rowland, M., White, G., and Karlen, E. (1984).
\newblock {Use of pellet-group plots to measure trends in deer and elk
  populations}.
\newblock {\em Wildlife Society Bulletin} {\bf 12,} 147--155.

\bibitem[\protect\citeauthoryear{Royen}{Royen}{2004}]{Royen:2004}
Royen, T. (2004).
\newblock Multivariate {G}amma distributions {II}.
\newblock In {\em Encyclopedia of Statistical Sciences}, pages 419--425. New
  York: John Wiley \& Sons.

\bibitem[\protect\citeauthoryear{Rue and Held}{Rue and
  Held}{2005}]{Rue:Held:2005}
Rue, H. and Held, L. (2005).
\newblock {\em {G}aussian {M}arkov {R}andom {F}ields: {T}heory and
  Applications}.
\newblock Chapman \& Hall, London.

\bibitem[\protect\citeauthoryear{Rue, Martino, and Chopin}{Rue
  et~al.}{2009}]{R-INLA:1}
Rue, H., Martino, S., and Chopin, N. (2009).
\newblock Approximate {B}ayesian inference for latent {G}aussian models using
  integrated nested {L}aplace approximations (with discussion).
\newblock {\em Journal of the Royal Statistical Society B} {\bf 71,} 319--392.

\bibitem[\protect\citeauthoryear{Sivertsen, Åhman, Steyaert, Rönnegård,
  Frank, Segerström, Støen, and Skarin}{Sivertsen
  et~al.}{2016}]{SivertsenEtAl2016}
Sivertsen, T.~R., Åhman, B., Steyaert, S. M. J.~G., Rönnegård, L., Frank,
  J., Segerström, P., Støen, O.-G., and Skarin, A. (2016).
\newblock Reindeer habitat selection under the risk of brown bear predation
  during calving season.
\newblock {\em Ecosphere} {\bf 7,} e01583.

\bibitem[\protect\citeauthoryear{Skarin and Alam}{Skarin and
  Alam}{2017}]{SkarinEtAl2017}
Skarin, A. and Alam, M. (2017).
\newblock {Reindeer habitat use in relation to two small wind farms, during
  preconstruction, construction, and operation}.
\newblock {\em Ecology and Evolution} {\bf 7,} 3870--3882.

\bibitem[\protect\citeauthoryear{Skarin, Nelleman, R{\"o}nneg{\aa}rd,
  Sandstr{\"o}m, and Lundqvist}{Skarin et~al.}{2015}]{SkarinEtAl2015}
Skarin, A., Nelleman, C., R{\"o}nneg{\aa}rd, Sandstr{\"o}m, P., and Lundqvist,
  H. (2015).
\newblock {Wind farm construction impacts reindeer migration and movement
  corridors}.
\newblock {\em Landscape Ecology} {\bf 30,} 1527--1540.

\bibitem[\protect\citeauthoryear{Stein}{Stein}{1999}]{Stein:1999}
Stein, M. (1999).
\newblock {\em Interpolation of Spatial Data. Some Theory of Kriging}.
\newblock Springer-Verlag, New York.

\bibitem[\protect\citeauthoryear{Trivedi and Zimmer}{Trivedi and
  Zimmer}{2017}]{Trivedi:Zimmer:2017}
Trivedi, P. and Zimmer, D. (2017).
\newblock A note on identification of bivariate copulas for discrete count
  data.
\newblock {\em Econometrics} {\bf 5(1),} 1--11.

\bibitem[\protect\citeauthoryear{Varin, Reid, and Firth}{Varin
  et~al.}{2011}]{Varin:Reid:Firth:2011}
Varin, C., Reid, N., and Firth, D. (2011).
\newblock An overview of composite likelihood methods.
\newblock {\em Statistica Sinica} {\bf 21,} 5--42.

\bibitem[\protect\citeauthoryear{Vere-Jones}{Vere-Jones}{1997}]{VereJones:1997}
Vere-Jones, D. (1997).
\newblock Alpha-permanents and their applications to multivariate gamma,
  negative binomial and ordinary binomial distributions.
\newblock {\em New Zealand Journal of Mathematics} {\bf 26,} 125--149.

\bibitem[\protect\citeauthoryear{Virtanen, Gommers, Oliphant, {Haberland},
  Reddy, Cournapeau, Burovski, Peterson, Weckesser, Jonathan, van~der Walt, J.,
  Brett, Wilson, Jarrod~Millman, Mayorov, and et~al.}{Virtanen
  et~al.}{2020}]{2020SciPy-NMeth}
Virtanen, P., Gommers, R., Oliphant, T., {Haberland}, M., Reddy, T.,
  Cournapeau, D., Burovski, E., Peterson, P., Weckesser, W., Jonathan, B.,
  van~der Walt, J., S., Brett, M., Wilson, J., Jarrod~Millman, K., Mayorov, N.,
  and et~al. (2020).
\newblock {{SciPy} 1.0: Fundamental Algorithms for Scientific Computing in
  Python}.
\newblock {\em Nature Methods} {\bf 17,} 261--272.

\bibitem[\protect\citeauthoryear{Winkelmann}{Winkelmann}{1995}]{wink95}
Winkelmann, R. (1995).
\newblock Duration dependence and dispersion in count-data models.
\newblock {\em Journal of Business \& Economic Statistics} {\bf 13,} 467--474.

\end{thebibliography}

\end{document}